\documentclass[12pt]{article}
\usepackage[utf8]{inputenc}
\usepackage{eurosym}
\usepackage{amssymb}
\usepackage{amsfonts}
\usepackage{amsmath}
\usepackage[authoryear,round]{natbib}
\usepackage{float}
\usepackage[margin=.9in]{geometry}
\usepackage{soul}
\usepackage[hang]{caption}
\usepackage{color}
\usepackage{lscape}
\usepackage{rotating}
\usepackage{xcolor}
\usepackage{setspace}
\usepackage{geometry}
\usepackage{graphicx}
\usepackage{subfig}
\usepackage{multicol}
\usepackage{etoolbox}
\usepackage{relsize}
\usepackage{multirow}
\usepackage{lscape}
\usepackage{amsmath}
\usepackage{amssymb}
\usepackage{tabularx}
\usepackage{threeparttable}
\usepackage{xr-hyper}
\usepackage{hyperref}
\usepackage{caption}
\usepackage{booktabs}
\newtheorem{theorem}{Theorem}[section]

\newtheorem{assumption}{Assumption}

\newtheorem{definition}{Definition}
\newtheorem{example}[theorem]{Example}

\newtheorem{lemma}{Lemma}[section]

\newtheorem{proposition}{Proposition}[section]
\newtheorem{remark}{Remark}

\newenvironment{proof}[1][Proof]{\textbf{#1.} }{\ \rule{0.5em}{0.5em}}
\begin{document}

\title{Semiparametric Local Projections\thanks{%
We thank Jiaming Huang, Carlos Lamarche, Florian Gunsilius, seminar participants at Vanderbilt University, University of Kansas, Wake Forest University and Pittsburgh University as
well as participants at the IAAE 2023, CIREQ 2023, MEG 2024, XV Time Series Workshop 2025, CEF-CM Statisics 2025, and ICEEE 2025 conferences for helpful comments. Jung
Jae Kim and Cristhian Rosales-Castillo provided excellent research assistance. This paper is based on
research supported by the NSF under Grants No. SES-2417534 and SES-2417535 and by a Natural Sciences and Engineering Research Council of Canada (NSERC) Grant No. RGPIN-2021-02663.
The views expressed in this paper are those of the authors and do not
necessarily represent the views of the Federal Reserve Bank of Dallas or the
Federal Reserve System. Refine.ink was used to check the paper for consistency and clarity.} }
\author{S\'{\i}lvia Gon\c{c}alves\thanks{McGill University, Department of Economics, 855 Sherbrooke St. W., Montr\'{e}al, Qu\'{e}bec, H3A 2T7, Canada. E-mail: silvia.goncalves@mcgill.ca.} \and
Ana Mar\'{\i}a Herrera\thanks{University of Kentucky, Department of Economics, 550 South Limestone, Lexington, KY 40506-0034, USA. E-mail: amherrera@uky.edu.} \and
Lutz Kilian\thanks{Federal Reserve Bank of Dallas, Research Department, 2200 N. Pearl St., Dallas, TX 75201, USA. E-mail: lkilian2019@gmail.com.} \and
Elena Pesavento\thanks{Emory University, Economics Department, 1602 Fishburne Dr. Atlanta, GA 30322, USA. E-mail: epesave@emory.edu.} \and
Iones Kelanemer Holban\thanks{McGill University, Department of Economics, 855 Sherbrooke St. W., Montr\'{e}al, Qu\'{e}bec, H3A 2T7, Canada.}}

\maketitle

\begin{abstract}
We propose a semiparametric local projection estimator of nonlinear impulse 
response functions for a broad class of structural dynamic models 
relevant for applied macroeconomics, including models with nonlinearly 
transformed regressors, state dependent coefficients, and nonlinear 
interactions between shocks and state variables. The estimator is based on 
a doubly robust moment condition that identifies the average response 
function as a linear functional of a nonparametric conditional mean, 
augmented by a density ratio that captures the effect of shifting the shock 
of interest. We combine this moment condition 
with cross-fitting that handles serial dependence. The resulting 
estimator is $\sqrt{T}$-consistent and asymptotically normal. We examine the finite-sample performance of the estimator across 
a range of nonlinear data generating processes and illustrate its use in two empirical examples.
\

JEL codes: C14, C32, E52, Q43

Keywords: impulse response, local projection, semiparametric estimation,
double machine learning, nonlinear structural model, potential outcomes.
\end{abstract}

\doublespacing

\section{Introduction}

Impulse response analysis is a cornerstone of empirical macroeconomics.
Local projections have become a popular method for estimating impulse
response functions (IRFs). In their simplest form, local projections consist of a
sequence of OLS regressions, one for each horizon of interest. The impulse
response of interest may be recovered from the estimated regressions without
further transformations of the model coefficients or the need for Monte
Carlo integration methods.

A large empirical literature has used
generalizations of linear local projections to evaluate state-dependent
impulse responses and other nonlinear responses. However, as shown in
\citet{goncalves2021,goncalves2024b, goncalves2024a}, such state-dependent local projections fail to
recover the population responses when the state is endogenous and the shock is large as in much of applied work (e.g., \citet{ramey2018}).

In this paper, we propose a semiparametric local projection estimator of
nonlinear impulse response functions that is valid across a broad range of
nonlinear settings relevant for applied macroeconomists. These include
processes with nonlinearly transformed regressors
\citep{herrera2015, tenreyro2016, benzeev2023, caravello2024},
state-dependent coefficients \citep{ramey2018}, and nonlinear
interactions between shocks and state variables
\citep{caramp2026monetary, cloyne2020decomposing, cloyne2023statedependent}. 

As is common in macroeconomics, our object of interest is the impulse response function of an outcome variable $y_{t+h}$ with respect to the primitive structural shock $\varepsilon_{1t}$ in the equation for variable $x_t$. Specifically, we aim to identify and estimate the response of $y_{t+h}$ to a shock of size $\delta$ in $\varepsilon_{1t}$. We assume that $x_t$ is predetermined with respect to $y_t$, an exclusion restriction that encompasses situations in which $x_t = \varepsilon_{1t}$ is an observed i.i.d.\ shock, as in the narrative approach to identification. In this case, $x_t$ is unconditionally independent of all other shocks driving the system between $t$ and $t+h$. When $x_t$ is not an observed shock $\varepsilon_{1t}$, the exclusion restriction together with the i.i.d.\ assumption on the structural shocks implies that $x_t$ is conditionally independent of all other shocks between $t$ and $t+h$, given control variables $\mathbf{z}_{t-1}$ that include the history of the system up to $t-1$.

We show how this conditional independence condition, combined with the assumption that $\varepsilon_{1t}$ enters $x_t$ additively, can be used to identify the IRF of $y_{t+h}$ with respect to $\varepsilon_{1t}$ using only the observables $(y_{t+h},x_t,\mathbf{z}_{t-1})$. In particular, the additive structure ensures that a $\delta$-perturbation in $\varepsilon_{1t}$ translates into a $\delta$-perturbation in $x_t$, holding fixed the control variables. The structural equation for the outcome variable is left unrestricted, permitting arbitrary nonlinearities. 

This identification strategy places our problem within the
semiparametric literature on inference for linear functionals of regression
functions (\citet{newey1994, chernozhukov2018, chernozhukov2022}). The key
estimation challenge is that the conditional mean function
$g_{0,h}(x,z) = E(y_{t+h}|x_t=x, \mathbf{z}_{t-1}=z)$ must be estimated
nonparametrically, which can induce bias in the plug-in
estimator. We address this concern by using the doubly robust moment condition of
\citet{chernozhukov2022}, augmented by
a density ratio reflecting the relative change in the conditional
distribution of $x_t$ given $\mathbf{z}_{t-1}$ when $x_t$ is shifted by
$\delta$. 

To handle the serial dependence of the data, we combine this moment condition
with the NLO (``neighbors-left-out'') cross-fitting approach of
\citet{semenova2023}, which ensures approximate independence between
training and evaluation sets. We derive the asymptotic distribution of the
resulting estimator and show that it is $\sqrt{T}$-consistent and
asymptotically normal, with the preliminary estimation of the nuisance
functions having no effect on the first-order asymptotic distribution.

Our paper is related to a recent and growing literature on semiparametric
and nonparametric inference for IRFs. The problem
of estimating the average effects of policy interventions nonparametrically
dates back at least to \citet{stock1989}, but our focus is on the causal
effect of structural macroeconomic shocks. Several recent papers have proposed nonparametric methods for
estimating nonlinear IRFs. For example, \citet{gourieroux2023} propose a nonparametric
local projection estimator for nonstructural IRFs identified from Gaussian shocks within
a Markov process framework. \citet{ballarin2024} proposes a sieve-based nonparametric estimator for
IRFs in models with nonlinearly transformed regressors, as in our
Example~\ref{Exsignsize} below, but does not cover the doubly robust
approach, the more general nonlinear settings we consider, or
inference. While our paper deals with shocks of finite magnitude $\delta$, \citet{kolesar2025}
focus on infinitesimally small shocks. They discuss the causal content of linear local projections when the data generating process is nonlinear and highlight challenges in applying doubly robust methods in the small samples typical of macroeconomics. Our paper
builds on this work as well as two recent studies that apply doubly
robust methods in macroeconometrics. \citet{ballinari2025} develop
semiparametric inference for IRFs using double/debiased machine learning
in a time series context, but focus on a binary treatment, so the
adjustment term in their orthogonal moment condition is based on the
propensity score rather than the density ratio we employ for continuous
treatments. \citet{huang2026} develop a two-step high-dimensional
nonparametric local projection estimator combining Neyman-orthogonal
pseudo-outcomes with cross-fitting. Because they focus on an IRF that
shifts the policy variable from a fixed baseline, their second step
requires a nonparametric regression and yields convergence rates slower
than $\sqrt{T}$; in contrast, our estimand averages over the
distribution of the shock and can be estimated at rate $\sqrt{T}$.
Finally, and independently, \citet{Nikolaishvili2026}, building on an earlier
version of our paper \citep{goncalves2024b}, proposes a closely related
doubly robust estimator for nonparametric local projections under
different assumptions and without allowing for covariates in the
conditioning set.

While our main focus is on identifying unconditional responses, our analysis 
also extends to average response functions conditional on a state variable 
$\Omega_t$. We provide identification conditions for this object and show how 
our semiparametric local projections estimator can be applied to each subsample 
$\{t: \Omega_t = \omega\}$ when $\Omega_t$ is discrete, as in state-dependent 
models.

The paper is organized as follows. Section~\ref{sec:framework} introduces
the structural model and three leading examples. Section~\ref{sec:IRFs}
defines the population IRFs of interest and contrasts our definition with alternative definitions used in the literature.  Sections~\ref{sec:ident} and~\ref{sec:estimation}
discuss identification and estimation, and inference, respectively.
Section~\ref{sec:conditionalIRF} briefly discusses how to extend the analysis to conditional
IRFs. The simulation results are presented in
Section~\ref{sec:simulations}. Section~\ref{sec:empirical} contains two empirical illustrations focusing on possible nonlinearities in the pass-through of gasoline price shocks to inflation and in the response of motor vehicle sales to real gasoline price shocks. We conclude in
Section~\ref{sec:conclusion}. The proofs are relegated to Appendices~\ref{sec:proof_prop} and \ref{sec:proofs}.

\section{Framework}\label{sec:framework}
Let $z_t=(x_t, y_t)'$ denote a vector of observed time series, where
$y_t$ is the outcome of interest and $x_t$ is predetermined with respect to $y_t$.
For example, $y_t$ could be real GDP and $x_t$ government spending. For simplicity, we assume that $y_t$ is
univariate, but extensions to multivariate outcomes could be easily accommodated. A general structural model for $z_t$ is a triangular system
of the form
\begin{align}
  x_t &= \phi(\mathbf{z}_{t-1})+\varepsilon_{1t}, \label{eq:S1} \\
  y_t &= \mu(x_t,\, \mathbf{z}_{t-1}, \varepsilon_{2t}),  \label{eq:S2}
\end{align}
where $\phi$ and $\mu$ are (potentially unknown) nonlinear functions,
and $\varepsilon_t = (\varepsilon_{1t}, \varepsilon_{2t})'$ is a vector of mutually independent structural
shocks. We assume $\varepsilon_t$ to be i.i.d.\ with mean zero and diagonal
covariance matrix $\Sigma = \mathrm{diag}(\sigma_1^2, \sigma_2^2)$. The
vector $\mathbf{z}_{t-1}= (z_{t-1}, z_{t-2}, \ldots, z_{t-p})'$ contains
lags of $z_t$; other variables can be included in $\mathbf{z}_{t-1}$ provided they are predetermined with respect to $\varepsilon_{1t}$. The exclusion of $y_t$ from equation (\ref{eq:S1}) is equivalent
to the assumption of block recursiveness in linear structural VAR
identification, where $x_t=\phi'\mathbf{z}_{t-1}+\varepsilon_{1t}$. An
important special case is $x_t=\varepsilon_{1t}$, as in the narrative
approach to identification.

Consistent with standard impulse response analysis in macroeconomics,
our goal is to trace the effect over time of surprise changes in the
variable $x_t$, as captured by a one-time change in the
structural shock $\varepsilon_{1t}$, rather than changes in $x_t$ itself.
This is the main reason why we assume in (\ref{eq:S1}) that $x_t$ and
$\varepsilon_{1t}$ are separable, as this is crucial for identifying the
impulse response function of $y_{t+h}$ with respect to $\varepsilon_{1t}$ from the observables $(y_{t+h},x_t,\mathbf{z}_{t-1})$.
A non-separable specification $x_t=\phi(\mathbf{z}_{t-1},\varepsilon_{1t})$
could be considered if instead we targeted the impulse
response function of $y_{t+h}$ with respect to $x_t$, as in \citet{kolesar2025} and \citet{huang2026}. 

Our framework accommodates a range of models of interest in applied work. One example is a model with nonlinearly transformed regressors. This
model allows for a sign nonlinearity in the responses with the magnitude of
the response depending on the sign of $x_t$ (e.g., $f(x_t)=\max \{x_t,0\}$%
) or a size nonlinearity with the magnitude of the response depending on the
size of $x_t$ (e.g., $f(x_t)=x^3_t$). Although the regression model is
linear in the parameters, the impulse response function is nonlinear,
requiring the use of nonstandard estimation methods (e.g., \citet{kilian2011}, \citet{goncalves2021}). Models with nonlinearly
transformed regressors have been used extensively in applied macroeconomics.
Examples include studies of the asymmetry in the responses to positive and
negative oil price shocks (e.g., \citet{herrera2015}) as well as
nonlinearities in the response of GDP to monetary policy shocks (e.g.,
\citet{tenreyro2016}, \citet{ascari2022}), financial shocks
(e.g., \citet{forni2024}) and fiscal shocks (e.g., \citet{benzeev2023}).
\begin{example}[Model with Nonlinear Regressors]
\label{Exsignsize}Let\vspace{-0.2cm}
\begin{align*}
x_{t}&=\phi'\mathbf{z}_{t-1}+\varepsilon _{1t}, \\ 
y_{t}&=\beta x_{t}+\rho' \mathbf{z}_{t-1}+cf\left( x_{t}\right) +\varepsilon _{2t},%
\end{align*}
where $f$ is a potentially unknown nonlinear function.
\end{example}

A second example is the state-dependent model examined in \citet{goncalves2024a} in which the response is allowed to differ between two observed
states (e.g., expansion and recession) based on a dummy variable indicator $%
S_{t-1}$. Models of this type have been used extensively to study the
magnitude of the fiscal multiplier, the effectiveness of monetary policy,
and the impact of uncertainty shocks in expansions and recessions (e.g.,
\citet{ramey2018}, \citet{cacciatore2021}, \citet{falck2021}).

\begin{example}[State-Dependent Model]
\label{ex:state} \label{Exstate}Let \vspace{-0.5cm}
\begin{align*}
x_{t}&=\varepsilon _{1t} \\ 
y_{t}&=\beta _{t-1}x_{t}+\gamma _{t-1}'\mathbf{z}_{t-1}+\varepsilon _{2t},%
\end{align*}
\noindent with $\beta_{t-1}=\beta_{E}S_{t-1}+\beta_{R}(1-S_{t-1})$
and $\gamma_{t-1}=\gamma_{E}S_{t-1}+\gamma_{R}(1-S_{t-1})$, where $S_{t-1}$
is a dummy variable indicating whether the economy is in expansion or in
recession. When $S_{t-1}$ depends on $\mathbf{z}_{t-1}$, the endogenous variables in the system, the IRF becomes nonlinear in $\varepsilon_{1t}$.
\end{example}

A final example is inspired by \citet{cloyne2020decomposing, cloyne2023statedependent} and \citet{caramp2026monetary} who consider a model in which the responses of $y_{t+h}$ to $\varepsilon_{1t}$ are allowed to be heterogeneous, with the heterogeneity being captured by an observable
variable, say, $r_t$. For instance, imagine a situation in which monetary
policy shocks, $\varepsilon_{1t}$, have a heterogeneous effect on GDP
growth, $y_t$, that depends on the level of government debt, $r_{t}$. The
level of debt, in turn, is a function of monetary policy in the previous
period ($x_{t-1}$) through its effect on interest rates. The interaction between the debt level and the shock of interest induces a nonlinearity that needs to be taken into account when estimating the IRF. Note that this specification differs from the
state-dependent model discussed earlier in that the model coefficients do
not depend on the state, but the impulse response does.

\begin{example}[Nonlinear Interaction Term of Shock with State]
\label{ex:cloyne} Let \label{ex:inter} \vspace{-0.5cm}
\begin{align*}
x_{t}&=\varepsilon _{1t} \\ 
y_t&=\beta_{21}x_t+ \beta_{23}r_{t}+\alpha_{21} x_tr_{t} +\gamma_{21} y_{t-1}
+\varepsilon_{2t} \\ 
r_t&=f(x_{t-1})+ \varepsilon_{3t},
\end{align*}
where $ \varepsilon _{1t}$ is mutually independent of $
\varepsilon_{2t}$ and $\varepsilon_{3t}$, $
\varepsilon_{2t}$ and $\varepsilon_{3t}$ are potentially correlated, $x_t$
is the shock of interest, and $r_t$ is an observable variable that may
change the effect of the policy shock. The form of the function $f$ is unknown and
can be linear or nonlinear.
\end{example}

\section{Population impulse response functions}\label{sec:IRFs}
Our main analysis focuses on unconditional versions of the IRF. Although the estimation and inference methods presented in Section~\ref{sec:estimation} are tailored to estimating unconditional IRFs, they can also be applied to conditional IRFs in state-dependent models such as in Example~\ref{ex:state}, where the conditioning set is discrete. Section~\ref{sec:conditionalIRF} discusses this application as well as the challenges one would face in estimating conditional IRFs in other examples such as Example~\ref{ex:inter}. 

Following the standard approach in macroeconomics, we care about the response of $y_{t+h}$ with respect to the structural shock $\varepsilon_{1t}$. As in the recent macroeconometrics literature, we adopt a potential outcomes framework (see e.g., \citet{goncalves2021} and \citet{goncalves2024a}). One implication of the structural model \eqref{eq:S1}--\eqref{eq:S2} is that $y_{t+h}$
can be written as $y_{t+h} = m_h(\varepsilon_{1t},\, U_{t+h})$, where $m_h$ is obtained by iterating \eqref{eq:S2} forward $h$
steps and substituting \eqref{eq:S1}, and
$U_{t+h} \equiv (\varepsilon_{2t},\, \varepsilon_{1,t+1},\, \varepsilon_{2,t+1},\,
\ldots,\, \varepsilon_{1,t+h},\, \varepsilon_{2,t+h},\, \mathbf{z}'_{t-1})'$
collects all remaining determinants of $y_{t+h}$. Since $\varepsilon_{1t}$ is
i.i.d.\ and independent of $\{\varepsilon_{2t}\}$, $\varepsilon_{1t}$ is independent of $U_{t+h}$, which we write as $\varepsilon_{1t} \perp U_{t+h}$.
The potential outcome
associated with fixing $\varepsilon_{1t} = e$ is
$y_{t+h}(e) = m_h(e, U_{t+h})$, where $e$ is any fixed value in the support of $\varepsilon_{1t}$. The observed outcome satisfies 
$y_{t+h} = y_{t+h}(\varepsilon_{1t})$, implying that it is the value that we observe when $e$ takes the value $\varepsilon _{1t}$ that generated the observed data. The fact that $\varepsilon _{1t}$ and $U_{t+h}$ are mutually independent implies that the potential outcomes are
independent of $\varepsilon _{1t}$. 

To define the response function of  $y_{t+h}$ with respect to $\varepsilon_{1t}$, we compare the (observed)
baseline value $y_{t+h}(\varepsilon_{1t})$ with the
counterfactual (unobserved) value of $y$ at $t+h$ that would have been observed if $
\varepsilon_{1t}$ had been subject to a shock of size $\delta$, denoted $%
y_{t+h}(\varepsilon_{1t}+\delta)$ (e.g., \citet{potter2000}). In
particular, following \citet{goncalves2024a}, we adopt the following
definition:
\begin{definition}\label{def:CAR}The average response function of $y_{t+h}$ to a shock of size $\delta$ in $%
\varepsilon_{1t}$ is defined as $\mathrm{ARF}_h(\delta)\equiv E\left( y_{t+h}\left( \varepsilon
_{1t}+\delta \right) -y_{t+h}\left( \varepsilon _{1t}\right) \right)$.
\end{definition}

$\mathrm{ARF}_h(\delta)$ corresponds to the unconditional average response used in
\citet{goncalves2021}. A conditional average response function can also be defined as in Definition~\ref{def:CAR_1} in Section~\ref{sec:conditionalIRF}, following \citet{goncalves2024b}.\footnote{Alternatively, one could also define versions of these IRFs where $\delta\to 0$, as discussed in \citet{goncalves2024b}. In this
paper, we focus on responses to a shock of finite magnitude $\delta$, as is
common in applied work (e.g., \citet{ramey2018}).}

Definition \ref{def:CAR} is not the only possible definition of an unconditional IRF. Other studies such as \citet{koop1996}, \cite{rambachan2021common}, and \citet{huang2026}, for example, have instead compared the two potential
outcomes $y_{t+h}(e^\prime)$ and $y_{t+h}\left(e\right)$, often setting $%
e^{\prime}=\delta$ and $e=0$, which yields the alternative definition:

\begin{definition}
\label{def:RF}The response function of 
$y_{t+h}$ to a shock of size $\delta$ in $\varepsilon_{1t}=e$ is defined as
$\mathrm{ARF}^*_{h}( \delta, e )=E(y_{t+h}(e+\delta )-y_{t+h}(e))$.
\end{definition}

Whereas Definition \ref{def:RF} has been used widely in the literature,
Definition \ref{def:CAR} is more recent (\citet{goncalves2021, goncalves2024b, goncalves2024a}). 
These two definitions are equivalent when the potential outcome is linear
in $e$ for all horizons, as would be the case for a linear model or in special cases of Example \ref{Exstate} and
Example \ref{ex:cloyne} when the conditioning sets ($S_{t-1}$ and $r_t$,
respectively) are exogenous. However, in general, the two definitions
differ.

\begin{figure}[H]
\centering
\caption{Alternative IRF definitions} \label{fig:irf_def}
\subfloat[$f(x)=\max(x_t,0)$]{\includegraphics[scale=0.35]{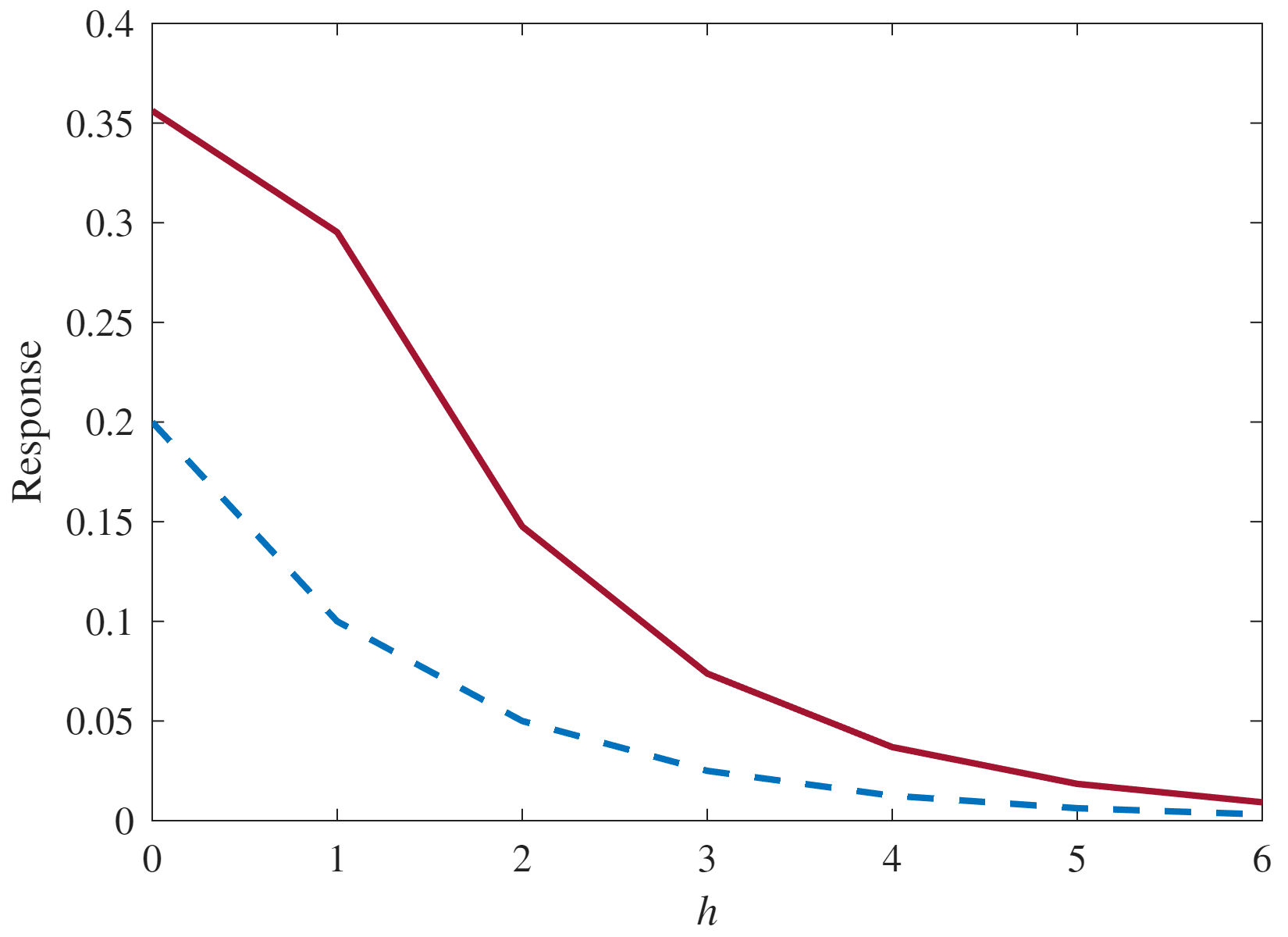}}
\subfloat[$f(x)=x^3_t$]{\includegraphics[scale=0.35]{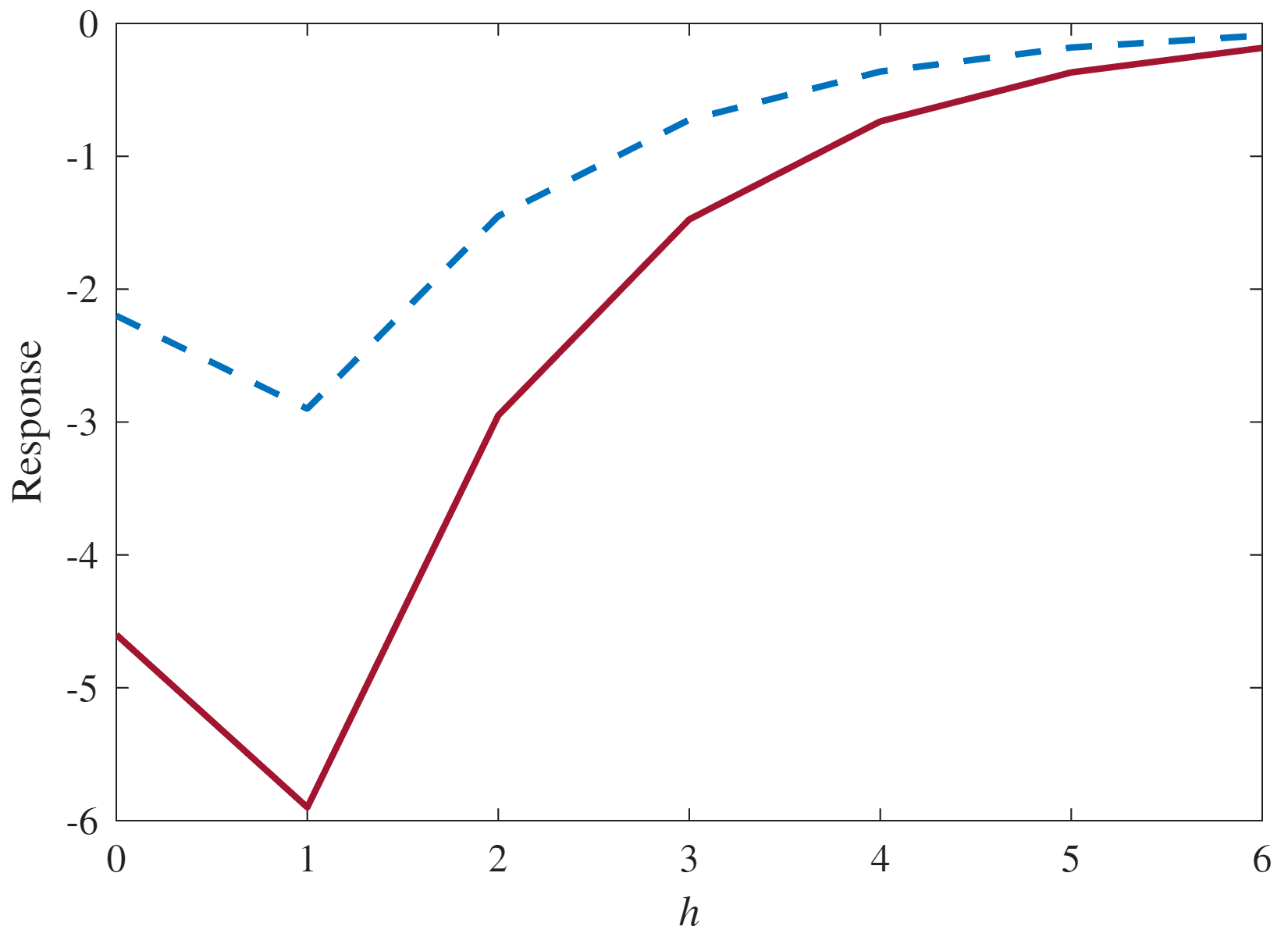}}
\\
\footnotesize{Notes: The solid red line and the dashed blue line correspond to the two definitions of the average response function $\mathrm{ARF}$ and $\mathrm{ARF}^*$ respectively, to a shock of size $\delta=2$.}
\end{figure}

Figure \ref{fig:irf_def} illustrates these differences by example. Consider
the nonlinear DGP: \vspace{-0.5cm}
\begin{eqnarray*}  \label{DGP12}
x_{t} &=&\varepsilon _{1t} \\
y_{t} &=&0.5y_{t-1}+0.5x_{t}+0.3x_{t-1}-0.4f(x_t)-0.3f(x_{t-1})+\varepsilon_{2t},
\end{eqnarray*}
where $\varepsilon_{1t}$ and $\varepsilon_{2t}$ are independent and have a
standard normal distribution. For illustrative purposes, let the magnitude of the shock be $\delta=2$ and
the functional forms $f(x_t)=\max(x_t,0)$ and $f(x_t)=x^3_t$,
respectively. The solid red line in Figure \ref{fig:irf_def} denotes the $\mathrm{ARF}$ obtained as
the average over the $y_{t}$ obtained for different realizations of $%
\varepsilon_{1t}$, whereas the dashed line denotes the value of $%
\mathrm{ARF}^{*}$ obtained by setting $\varepsilon_{1t}=e=0$. It is readily apparent
that in this example, the two definitions of the IRF imply
quite different measures of the conditional expectation of $y_{t}$ in the
absence of a perturbation.

Which approach is the more natural one? The only difference between these
two approaches is the treatment of the impact period. The baseline in
computing any impulse response is the conditional expectation of $y_{t}$ in
the absence of a perturbation $\delta$ (e.g., \citet{potter2000}, p.~1430). In other words, the baseline is what we would have expected $y_{t}$ to be in the absence
of a perturbation, possibly conditional on the history of the data. For example, if $\mathcal{F}^{t-1}$ denotes the information available up to time $t-1$, a natural baseline is the conditional expectation $E_{t-1}\left(y_t\right) \equiv E\left(y_t \mid \mathcal{F}^{t-1}\right)$. In this example, $
E_{t-1}(y_{t}) = 0.5y_{t-1}+0.3x_{t-1}-0.4E_{t-1}(f(x_t))-0.3f(x_{t-1}),$
where the predetermined values are known and we imposed $E_{t-1}(\varepsilon
_{1t})=E_{t-1}(\varepsilon
_{2t})=0.$ This expectation can only be evaluated by integrating $f(x_t)$ over all possible realizations of $x_t$, as in Definition 1. In contrast, Definition 2 evaluates this expression as $f(E(x_t))=f(0)$. By Jensen's inequality, this will not yield the desired baseline for computing the population impulse response to a shock of magnitude $\delta$ because $
E(f(x_t))$ is not $f(E(x_t))$. Thus, we work with Definition 1
throughout this paper.

\section{Identification}\label{sec:ident}
We discuss the identification of $\theta_{0,h}\equiv \mathrm{ARF}_h(\delta)$ in 
Definition~\ref{def:CAR}. Motivated by Section 6 of \citet{kolesar2025}, we first discuss identification based on a 
regression-based approach, where $\theta_{0,h}$ is identified using the 
conditional expectation function $g_{0,h}(x,z)\equiv E(y_{t+h}|x_t=x,
\mathbf{z}_{t-1}=z)$, and then show how to obtain identification using a 
doubly robust approach. We use the subscript ``0'' to indicate true 
parameters and functions throughout.

Starting with the regression-based approach, note that the independence 
between $\varepsilon_{1t}$ and $U_{t+h}$ (which holds by \eqref{eq:S1} 
and \eqref{eq:S2} under the i.i.d.\ assumption on $\varepsilon_t$ and the mutually independent shocks assumption) allows us to identify $\theta_{0,h}$ as 
$\theta_{0,h} = E\!\left[g_{\varepsilon,0,h}(\varepsilon_{1t}+\delta) - 
g_{\varepsilon,0,h}(\varepsilon_{1t})\right]$,
where $g_{\varepsilon,0,h}(e)\equiv E(y_{t+h}|\varepsilon_{1t}=e)$ (e.g., \citet{kolesar2025}). This representation is useful for estimation when 
$\varepsilon_{1t}$ is an observed shock, as in the narrative approach to 
identification. However, it does not directly apply when $x_t$ is an observed variable and $\varepsilon_{1t}$ is its underlying (unobserved) structural shock. In what follows, we show how the additive structure of 
\eqref{eq:S1} can be exploited to identify the average response function of $y_{t+h}$ to an impulse in $\varepsilon_{1t}$ (i.e., $\mathrm{ARF}_h(\delta)$ in Definition~\ref{def:CAR}) in this more general context: since $x_t=\phi(\mathbf{z}_{t-1})+\varepsilon_{1t}$, a 
$\delta$-shift in $\varepsilon_{1t}$ holding $\mathbf{z}_{t-1}$ fixed is 
identical to a $\delta$-shift in $x_t$ holding $\mathbf{z}_{t-1}$ fixed, 
yielding an identification result for $\theta_{0,h}$ in terms of the 
observables $(y_{t+h},x_t,\mathbf{z}_{t-1})$ only.

\begin{proposition}\label{Prop:identification}
Suppose that $z_t=(x_t, y_t)'$ satisfies \eqref{eq:S1} and \eqref{eq:S2} 
where $\varepsilon_t=(\varepsilon_{1t},\varepsilon_{2t})'$ contains mutually independent shocks and is
i.i.d.$(0,\Sigma)$ with $\Sigma=\mathrm{diag}(\sigma^2_1,\sigma^2_2)$. 
It follows that
\begin{equation}
  \theta_{0,h}\equiv\mathrm{ARF}_h(\delta) = E\!\left[g_{0,h}(x_t+\delta,\, 
  \mathbf{z}_{t-1}) - g_{0,h}(x_t,\, \mathbf{z}_{t-1})\right],
  \label{eq:ident} 
\end{equation}
where $g_{0,h}(x, z) \equiv E(y_{t+h} \mid x_t = x,\, \mathbf{z}_{t-1} = z)$.
\end{proposition}

The representation of $\mathrm{ARF}_h(\delta)$ given in \eqref{eq:ident} 
is identified from observables because $g_{0,h}(x,z)$ is the conditional 
mean of $y_{t+h}$ given $(x_t,\mathbf{z}_{t-1})$, a functional of the 
joint distribution of $(y_{t+h}, x_t, \mathbf{z}_{t-1})$, and the outer 
expectation is taken over the marginal distribution of $(x_t, 
\mathbf{z}_{t-1})$, both of which are observable. No knowledge of $\phi$ 
or $\mu$ is required.  Identification requires that $g_{0,h}$ be defined on the support of $(x_t+\delta,\mathbf{z}_{t-1})$ as well as on the support of $(x_t,\mathbf{z}_{t-1})$. When the support of $\varepsilon_{1t}$ (and hence of $x_t$ given $\mathbf{z}_{t-1}$) is bounded, the shifted support may fall outside the original one, in which case $g_{0,h}(x_t+\delta,\mathbf{z}_{t-1})$ is not identified for some $(x_t,\mathbf{z}_{t-1})$ pairs. For this reason, we do not impose bounded support conditions. 
\begin{remark}
The additive structure in \eqref{eq:S1} is important for identifying 
$\theta_{0,h}\equiv \mathrm{ARF}_h(\delta)$ as stated in Definition~\ref{def:CAR}. It implies that 
perturbing $\varepsilon_{1t}$ by $\delta$ is equivalent to perturbing 
$x_t$ by $\delta$ holding $\mathbf{z}_{t-1}$ fixed, so that $\theta_{0,h}$ 
can equivalently be written as
\begin{equation*}
\theta_{0,h}= E\!\left[m_h(x_t+\delta,\, U_{t+h}) - m_h(x_t,\, 
U_{t+h})\right]= E\!\left[g_{0,h}(x_t+\delta,\, \mathbf{z}_{t-1}) - 
g_{0,h}(x_t,\, \mathbf{z}_{t-1})\right],
\end{equation*} 
where $m_h$ is a reparametrization of the structural function introduced in Section~3, now expressed 
as a function of $x_t$. Without the additive structure, the first equality 
fails. Nevertheless, the estimand $E\!\left[g_{0,h}(x_t+\delta,\, \mathbf{z}_{t-1}) - 
g_{0,h}(x_t,\, \mathbf{z}_{t-1})\right]$ retains a causal interpretation 
as the IRF of $y_{t+h}$ to a $\delta$-perturbation of $x_t$ holding 
$\mathbf{z}_{t-1}$ fixed.
\end{remark}
\begin{remark}
If contemporaneous feedback from $y_t$ to $x_t$ were allowed, Definition~\ref{def:CAR} would still define the causal effect of a shock on the relevant potential outcome. However, the identification argument in Proposition~\ref{Prop:identification} would no longer apply directly: it would require an additional instrument or identification strategy to isolate the exogenous component of $x_t$. We reserve this question for future research.
\end{remark}

Proposition~\ref{Prop:identification} leads to a moment 
condition of the form
\begin{equation}\label{eq:naive_moment}
  E\left[g_{h}(x_t+\delta,\mathbf{z}_{t-1})-g_{h}(x_t,\mathbf{z}_{t-1})
  -\theta_h\right]=0,
\end{equation}
which identifies $\theta_{0,h}$ when $g_h=g_{0,h}$. A natural approach is to
replace $g_{0,h}$ by a first-step estimator $\hat{g}_h$, yielding 
$\widehat{\mathrm{ARF}}_h(\delta)=T^{-1}\sum_{t=1}^{T}\left[\hat{g}_h(x_t+\delta,
\mathbf{z}_{t-1})-\hat{g}_h(x_t,\mathbf{z}_{t-1})\right]$. When 
$\hat{g}_h$ is estimated by machine learning or nonparametric methods, 
this regression-based estimator can suffer from 
first-order bias, and inference requires adjusting for estimation 
uncertainty in $g_{0,h}$. This motivates using the double/debiased machine 
learning approach of \citet{chernozhukov2018, chernozhukov2022, chernozhukov2024covariate} to augment 
the moment condition \eqref{eq:naive_moment} using Neyman orthogonality. When combined with a form of cross-fitting that handles time series dependence, inference based on this estimator can proceed as if the nuisance functions were fully observed. 

To describe the doubly robust approach, let $f_{0,x|z}(x|z)$ denote the 
conditional density of $x_t$ given $\mathbf{z}_{t-1}$. It can be easily shown that for any function $g_h$,
\begin{equation}\label{eq:Neyman}
E\left[g_h(x_t+\delta,\mathbf{z}_{t-1})-g_h(x_t,\mathbf{z}_{t-1})\right]
=E\left[\alpha_0(x_t,\mathbf{z}_{t-1})g_h(x_t,\mathbf{z}_{t-1})\right],
\end{equation}
where $\alpha_0(x,z)\equiv (f_{0,x|z}(x-\delta|z)-f_{0,x|z}(x|z))/
f_{0,x|z}(x|z)$ is the Riesz representer (we omit the dependence on 
$\delta$ throughout). This Riesz representer is a density ratio that captures the relative change in the conditional 
density of $x_t$ given $\mathbf{z}_{t-1}$ when $x_t$ is shifted by $\delta$ (see Section 6 of \citet{kolesar2025}, who report the form of the Riesz representer when $\delta\to 0$).

\begin{proposition}\label{Prop:DR}
Suppose the conditions of Proposition~\ref{Prop:identification} hold. 
Then $\theta_{0,h}$ solves
\begin{equation}\label{eq:moment_DR}
E\left[g_{0,h}(x_t+\delta,\mathbf{z}_{t-1})-g_{0,h}(x_t,\mathbf{z}_{t-1})
-\theta_h+\alpha_0(x_t,\mathbf{z}_{t-1})(y_{t+h}
-g_{0,h}(x_t,\mathbf{z}_{t-1}))\right]=0.
\end{equation}
\end{proposition}

As it turns out, Equation \eqref{eq:moment_DR} yields a doubly robust moment equation for $\theta_{0,h}$. Defining $$\psi(y_{t+h},x_t,\mathbf{z}_{t-1},g_h,\alpha,\theta_{h})=g_{h}(x_t+\delta,\mathbf{z}_{t-1})-g_{h}(x_t,\mathbf{z}_{t-1})
-\theta_{h}+\alpha(x_t,\mathbf{z}_{t-1})(y_{t+h}
-g_{h}(x_t,\mathbf{z}_{t-1})),$$ we can show that $E[\psi(y_{t+h},x_t,\mathbf{z}_{t-1},g_h,\alpha,\theta_{0,h})]=0$ holds for
any $g_h$ when $\alpha=\alpha_0$ and for any $\alpha$ 
when $g_h=g_{0,h}$. This follows by an application of Theorem 5 of \cite{chernozhukov2022}. Hence, this moment condition yields an estimator that is insensitive to 
misspecification of $g_h$ provided $\alpha=\alpha_0$, and insensitive to 
misspecification of $\alpha$ provided $g_h=g_{0,h}$.
\section{Estimation and inference}\label{sec:estimation}

We estimate $\theta_{0, h}$ by combining the doubly robust moment condition in \eqref{eq:moment_DR} with cross-fitting. In the standard cross-fitting approach, the data are partitioned into non-overlapping blocks. For each block, the nuisance functions are estimated on the complement of that block and evaluated on the heldout observations. These out-of-sample nuisance estimates are then used to construct the moment condition for estimating $\theta_{0, h}$; see, for example, \citet{chernozhukov2022}. A crucial assumption that justifies this approach is random sampling, which implies that the blocks are mutually independent (and the data within blocks are i.i.d.). 

In a time series context, the presence of serial dependence in $z_{t}=(x_t,y_t)'$ violates this assumption even when $x_t$ is i.i.d., creating dependence among the blocks. Hence, we follow \citet{semenova2023} and rely on NLO (``neighbors-left-out'') cross-fitting.\footnote{This approach leaves out not only the target block but also its immediate neighbors when estimating the nuisance functions. The resulting training and evaluation sets are then approximately independent, with the approximation error controlled by the speed of mixing of the underlying time series. Recent applications of the NLO cross-fitting approach to inference on impulse response functions in time series include \citet{ballinari2025} and \citet{huang2026}. Because their estimands are different than ours, their assumptions and asymptotic results also differ from ours.} We partition the index set 
$\{1,\ldots,T\}$ into $K\ge 4$ non-overlapping blocks $I_1,\ldots,I_K$ 
of contiguous time indices, each of size $T_\ell\equiv T/K$: 
$\{1,\ldots,T\} = I_1\cup\cdots\cup I_K$.\footnote{For simplicity we assume that $T$ is divisible by $K$.} We keep $K$ fixed as $T\to \infty$ when deriving the asymptotic theory below. For each $\ell\in\{1,\ldots,K\}$, 
let $\mathcal{N}(\ell)$ denote the set containing $\ell$ and its immediate 
neighbors in $\{1,\ldots,K\}$, i.e.\ 
$\mathcal{N}(\ell)=\{\ell-1,\ell,\ell+1\}\cap\{1,\ldots,K\}$, and define 
the quasi-complement of $I_\ell$ as $I^{\mathrm{qc}}_\ell = \bigcup_{j\notin\mathcal{N}(\ell)} I_j,$ so that $I^{\mathrm{qc}}_\ell$ is obtained from the full complement 
$I_{-\ell}=\bigcup_{j\ne\ell}I_j$ by additionally removing the two blocks 
adjacent to $I_\ell$. The key feature of NLO cross-fitting is that 
$I_\ell$ and $I^{\mathrm{qc}}_\ell$ are separated by at least $T_\ell$ 
time periods. Since $K$ is fixed, as $T\to\infty$ we have $T_\ell\to\infty$, 
so that under mixing-type conditions $I_\ell$ is approximately independent 
of $I^{\mathrm{qc}}_\ell$.

Given horizon $h$, the NLO cross-fitting estimator of $\theta_{0,h}$, 
which we refer to as DR-NLO, is computed as follows. For notational 
simplicity, we drop the horizon index $h$ in the nuisance functions 
and estimator that correspond to block $\ell$, with the understanding that all objects depend on $h$. 
For each $\ell=1,\ldots,K$:
\begin{enumerate}
    \item Estimate $\hat{g}_\ell$ and $\hat{\alpha}_\ell$ using some 
    nonparametric or machine learning procedure on the quasi-complement 
    $I^{\mathrm{qc}}_\ell$.
    \item Set the average of the moment condition \eqref{eq:moment_DR} 
    over $I_\ell$ to zero to obtain
    \begin{equation*}
    \hat{\theta}_{\ell} = \frac{1}{T_\ell}\sum_{t\in I_\ell}\left[
    \hat{g}_\ell(x_t+\delta,\mathbf{z}_{t-1})-\hat{g}_\ell(x_t,
    \mathbf{z}_{t-1})+\hat{\alpha}_\ell(x_t,\mathbf{z}_{t-1})(y_{t+h}-
    \hat{g}_\ell(x_t,\mathbf{z}_{t-1}))\right].
    \end{equation*}
    \item Compute the estimator of $\theta_{0,h}$ as
    \begin{equation*}\label{eq:estimator}
    \hat{\theta}_h= \frac{1}{K}\sum_{\ell=1}^{K}\hat{\theta}_{\ell} 
    = \frac{1}{T}\sum_{\ell=1}^{K}\sum_{t\in I_\ell}\left[
    \hat{g}_\ell(x_t+\delta,\mathbf{z}_{t-1})-\hat{g}_\ell(x_t,
    \mathbf{z}_{t-1})+\hat{\alpha}_\ell(x_t,\mathbf{z}_{t-1})(y_{t+h}-
    \hat{g}_\ell(x_t,\mathbf{z}_{t-1}))\right].
    \end{equation*}
\end{enumerate}

We next state the regularity conditions used to derive the asymptotic distribution of $\hat{\theta}_h$. Following \citet{semenova2023} and \citet{huang2026}, we impose 
geometric $\beta$-mixing on $\{z_t\}$. $\beta$-mixing is 
strictly stronger than the more standard strong mixing assumption, but allows us to use the Strassen coupling result that underlies the NLO 
theory of \citet{semenova2023}. We define the $\beta$-mixing 
coefficients as $\beta(j)\equiv \sup_{t}\,\beta\!\left(\sigma(z_s,
s\le t),\,\sigma(z_s,s\ge t+j)\right)$, where
$\beta(\mathcal{A},\mathcal{B}) = E[\sup_{B\in\mathcal{B}}\left|P(B\mid
\mathcal{A})-P(B)\right|]$ for any $\sigma$-algebras
$\mathcal{A}$ and $\mathcal{B}$. The function $\psi(y_{t+h},x_t,\mathbf{z}_{t-1},g_0,\alpha_0,\theta_{0,h})$ in Assumption~\ref{Assumption:boundedness} below is the doubly robust moment function defined in \eqref{eq:moment_DR} (see Proposition~\ref{Prop:DR}). We let $\mathcal{X}$ and $\mathcal{Z}$ denote the supports of $x_t$ and $\mathbf{z}_{t-1}$.

\begin{assumption}\label{Assumption:mixing}
$\{z_{t}=(x_t, y_t)'\}$ is stationary and geometrically $\beta$-mixing, 
i.e., $\beta(j) \le C\exp(-c_\beta j)$ for some constants $C\ge 0$ and 
$c_\beta>0$, and for all $j\ge 1$.
\end{assumption}

\begin{assumption}\label{Assumption:boundedness}
For some finite constants $\bar{\alpha}$, $\bar{\sigma}_q$, and for some 
$q>2$:
\begin{itemize}
    \item[(i)] $\sup_{(x,z)\in\mathcal{X}\times\mathcal{Z}} |\alpha_0(x,z)| 
  < \bar{\alpha}$.
  
  \item[(ii)] $\sup_{(x,z)\in\mathcal{X}\times\mathcal{Z}} 
  E[|e_{t+h}|^q|x_t=x,\mathbf{z}_{t-1}=z] < \bar{\sigma}_q^q$, where 
  $e_{t+h}\equiv y_{t+h}-g_{0}(x_t,\mathbf{z}_{t-1})$.
  \item[(iii)] $E[|\psi(y_{t+h},x_t,\mathbf{z}_{t-1},g_{0},\alpha_0,\theta_{0,h})|^{2+
  \epsilon}]<\infty$ for some $\epsilon>0$.
\end{itemize}
\end{assumption}

Assumption~\ref{Assumption:boundedness}(i) requires the Riesz representer 
$\alpha_0(x,z)$ to be uniformly bounded over $\mathcal{X}\times\mathcal{Z}$, the support of 
$(x_t,\mathbf{z}_{t-1})$. This rules out distributions with thin-tails such as the Gaussian, but allows for Student-$t$ distributions. Assumption~\ref{Assumption:boundedness}(ii) imposes 
a conditional $q$th moment bound on the regression residual $e_{t+h}$ for 
$q>2$, which is needed to control autocovariance terms arising from serial 
dependence. Assumption~\ref{Assumption:boundedness}(iii) is a standard 
moment condition on the influence function that, together with 
Assumption~\ref{Assumption:mixing}, ensures that a central limit theorem 
applies to the oracle estimator based on the true nuisance functions.

\begin{assumption}\label{Assumption:learners}
For each $\ell=1,\ldots,K$, $\hat{g}_\ell\in\mathcal{G}_T$ and 
$\hat{\alpha}_\ell\in\mathcal{A}_T$ with probability converging to one, 
where $\mathcal{G}_T$ and $\mathcal{A}_T$ denote shrinking neighborhoods 
of $g_{0,h}$ and $\alpha_0$, respectively, and $q$ is as in 
Assumption~\ref{Assumption:boundedness}. In addition,
\begin{itemize}
  \item[(i)] $r_{g,q,T}\equiv \sup_{g\in\mathcal{G}_T}\bigl(E\bigl[|g
  (x_t,\mathbf{z}_{t-1})-g_{0,h}(x_t,\mathbf{z}_{t-1})|^{q}\bigr]
  \bigr)^{1/q}=o(1)$;\\
  $r_{\alpha,q,T}\equiv \sup_{\alpha\in\mathcal{A}_T}\bigl(E
  \bigl[|\alpha(x_t,\mathbf{z}_{t-1})-\alpha_0(x_t,\mathbf{z}_{t-1})|^{q}
  \bigr]\bigr)^{1/q}=o(1)$.
  \item[(ii)] $\sqrt{T}\, r_{g,T}\, r_{\alpha,T} = o(1)$, where 
  $r_{g,T}\equiv r_{g,2,T}$ and $r_{\alpha,T}\equiv r_{\alpha,2,T}$ 
  denote the $L_2$ convergence rates.
\end{itemize}
\end{assumption}
Assumption~\ref{Assumption:learners}(i) imposes $L_q$ consistency on both 
nuisance estimators, which implies $L_2$ consistency since $q>2$. 
Assumption~\ref{Assumption:learners}(ii) is the product rate condition 
standard in the double machine learning literature; see Assumption 2 of \cite{chernozhukov2024covariate}. The $L_q$ rates in (i) are stronger than 
the $L_2$ rates typically assumed in the i.i.d.\ case and are needed here 
to control autocovariance terms arising from serial dependence in $e_{t+h}$ 
and $(x_t,\mathbf{z}_{t-1})$; see Remarks~\ref{rem:autocovariance} and 
\ref{rem:iid} below.

\begin{theorem}\label{Theorem:NLO2}
Suppose that Assumptions~\ref{Assumption:mixing}, 
\ref{Assumption:boundedness}, and \ref{Assumption:learners} hold. Then
\begin{equation*}
\sqrt{T}(\hat{\theta}_h-\theta_{0,h}) = \frac{1}{\sqrt{T}}\sum_{t=1}^{T}
\psi(y_{t+h},x_t,\mathbf{z}_{t-1},g_{0,h},\alpha_0,\theta_{0,h}) + o_p(1) \to_d N(0,V_h),
\end{equation*}
where $V_h$ is the long-run variance of $\psi(y_{t+h},x_t,\mathbf{z}_{t-1},g_{0,h},\alpha_0,\theta_{0,h})$.
\end{theorem}

The proof of Theorem~\ref{Theorem:NLO2} is in Appendix~\ref{sec:proofs}. Letting $\tilde{\theta}_{h}$ denote the oracle estimator which assumes we know $g_0$ and $\alpha_0$, we decompose
\[
\sqrt{T}(\hat{\theta}_h-\theta_{0,h}) = \sqrt{T}(\tilde{\theta}_h-\theta_{0,h})+ \sqrt{T}(\hat{\theta}_h-\tilde{\theta}_h),
\]and show that $\sqrt{T}(\hat{\theta}_h-\tilde{\theta}_h)\equiv \sqrt{T}(R_1+R_2+R_3)=o_p(1)$, where $R_1$, $R_2$ and $R_3$ are remainder terms defined in Appendix~\ref{sec:proofs}. Assumptions~\ref{Assumption:mixing}, \ref{Assumption:boundedness} and \ref{Assumption:learners} suffice for proving that these remainders are $o_p(T^{-1/2})$.
\begin{remark}\label{rem:autocovariance}
Assumption~\ref{Assumption:boundedness}(ii) and the $L_q$ rate on 
$\hat{\alpha}_\ell$ in Assumption~\ref{Assumption:learners}(i) can both 
be relaxed to $q=2$, as is standard in the i.i.d.\ case, provided we 
impose a direct condition on the autocovariance structure of $e_{t+h}$. 
For instance, it suffices to assume $E(e_{t+h}|x_t,\mathcal{F}^{t-1})=0$, where $\mathcal{F}^{t-1}$ denotes the lagged history of $z_t$, together with uniform boundedness of the conditional autocovariances. Alternatively, one can drop the martingale difference condition and instead assume absolute summability of the conditional autocovariance sequence. The $L_q$ rate on $\hat{g}_\ell$ in 
Assumption~\ref{Assumption:learners}(i) is still required, as it controls 
autocovariance terms in $R_1$ arising from serial dependence in 
$(x_t,\mathbf{z}_{t-1})$ rather than in $e_{t+h}$.
\end{remark}
\begin{remark}\label{rem:iid}
For the special case where $x_t$ is i.i.d.\ and $x_t\perp U_{t+h}$ 
without control variables $\mathbf{z}_{t-1}$, the $L_q$ rate on 
$\hat{g}_\ell$ in Assumption~\ref{Assumption:learners}(i) is not required. 
The main reason is that we can exploit the independence of $x_t$ when 
showing that $\sqrt{T}R_1=o_p(1)$. However, even when $x_t$ is i.i.d., 
$e_{t+h}$ may be serially dependent, so the $L_q$ rate on $\hat{\alpha}_\ell$ 
in Assumption~\ref{Assumption:learners}(i) is still needed for 
$\sqrt{T}R_2=o_p(1)$. Alternatively, only $L_2$ rates on $\hat\alpha_\ell$ 
are required if further conditions on the autocovariance structure of 
$e_{t+h}$ are imposed, as discussed in Remark~\ref{rem:autocovariance}.
\end{remark}

The main implication of Theorem~\ref{Theorem:NLO2} is that the preliminary estimation of 
$g_{0,h}$ and $\alpha_0$ does not affect the first-order asymptotic 
distribution of $\hat{\theta}_h$: the estimator is asymptotically 
equivalent to the infeasible oracle estimator that uses the true nuisance 
functions. This is a consequence of the Neyman orthogonality of the moment 
condition \eqref{eq:moment_DR} combined with the approximate independence 
between $I_\ell$ and $I^{\mathrm{qc}}_\ell$ under NLO cross-fitting. A feasible confidence interval for $\theta_{0,h}$ can be constructed using 
a standard HAC estimator of $V_h$ based on the estimated influence function 
$\psi(y_{t+h},x_t,\mathbf{z}_{t-1},\hat{g}_\ell,\hat{\alpha}_\ell,\hat{\theta}_h)$.
\section{Conditional impulse response functions}\label{sec:conditionalIRF}
Conditional impulse response functions are often of interest in applications 
such as in Examples~\ref{ex:state} or \ref{ex:inter}. A generalization of 
Definition~\ref{def:CAR} to conditional IRFs is as follows.

\begin{definition}\label{def:CAR_1}
The conditional average response function of $y_{t+h}$ to a shock of size 
$\delta$ in $\varepsilon_{1t}$ is defined as $CAR_h(\delta, \omega) \equiv 
E\left(y_{t+h}\left(\varepsilon_{1t}+\delta\right)-y_{t+h}\left(\varepsilon_{1t}
\right)| \Omega_{t}=\omega\right)$, where $\Omega_{t}$ denotes the conditioning 
set.
\end{definition}

Since $\varepsilon_{1t}$ is random, the conditional expectation in 
Definition~\ref{def:CAR_1} averages over all possible realizations of 
$\varepsilon_{1t}$ (in addition to the other sources of randomness that enter 
into the potential outcomes through $U_{t+h}$), conditionally on $\Omega_{t}=\omega$. 
The choice of $\omega$ in $CAR_h(\delta,\omega)$ is context-dependent. For 
instance, in Example~\ref{Exstate} the conditioning set $\Omega_{t}$ is the 
state variable at time $t-1$, i.e.\ $\Omega_{t}=S_{t-1}$, so $\omega$ is 
either 0 or 1, while in Example~\ref{ex:cloyne} the conditioning set is 
$\Omega_{t}=r_{t}$, so $\omega$ can take on any value in the support of 
$r_{t}$.

The following proposition provides conditions under which $CAR_h(\delta,\omega)$ 
is identified.
\begin{proposition}\label{Prop:identification_CAR}
Suppose that $z_t=(x_t,y_t)'$ satisfies \eqref{eq:S1} and \eqref{eq:S2} 
where $\varepsilon_t=(\varepsilon_{1t},\varepsilon_{2t})'$ is a vector of mutually independent shocks that is distributed
i.i.d.$(0,\Sigma)$ with $\Sigma=\mathrm{diag}(\sigma^2_1,\sigma^2_2)$. 
Suppose further that $x_t \perp U_{t+h} \mid (\mathbf{z}_{t-1}, \Omega_t)$.
It follows that
\begin{equation}
  \mathrm{CAR}_h(\delta, \omega) = E\!\left[g_{0,h}(x_t+\delta,\, 
  \mathbf{z}_{t-1}, \omega) - g_{0,h}(x_t,\, \mathbf{z}_{t-1}, \omega) 
  \mid \Omega_t = \omega\right],
  \label{eq:ident_CAR} 
\end{equation}
where $g_{0,h}(x, z, \omega) \equiv E(y_{t+h} \mid x_t = x,\, 
\mathbf{z}_{t-1} = z,\, \Omega_t = \omega)$.
\end{proposition}
The key condition in Proposition~\ref{Prop:identification_CAR} is 
$x_t \perp U_{t+h} \mid (\mathbf{z}_{t-1}, \Omega_t)$, which requires that 
conditionally on $(\mathbf{z}_{t-1}, \Omega_t)$, all remaining variation in 
$x_t$ comes from $\varepsilon_{1t}$ alone. This condition holds in two 
leading cases. First, if $\Omega_t$ is a function of $\mathbf{z}_{t-1}$ 
alone, then conditioning on $(\mathbf{z}_{t-1}, \Omega_t)$ is the same as 
conditioning on $\mathbf{z}_{t-1}$ alone, and conditional independence follows 
directly from \eqref{eq:S1} and \eqref{eq:S2}, as in 
Proposition~\ref{Prop:identification}. This covers Example~\ref{Exstate}, 
where $\Omega_t = S_{t-1}$ is a function of $\mathbf{z}_{t-1}$.\footnote{This 
corresponds to the case where the state of the economy depends on lags of the 
outcome of interest as it is the case when studying state-dependent government 
spending multipliers.} Second, if $\Omega_t$ contains variables not in 
$\mathbf{z}_{t-1}$, conditional independence is satisfied if $\Omega_t \perp 
\varepsilon_{1t}$. This covers Example~\ref{ex:cloyne}, where $\Omega_t = 
r_t = f(x_{t-1}) + \varepsilon_{3t}$: since $x_{t-1}$ is dated $t-1$ and 
$\varepsilon_{3t} \perp \varepsilon_{1t}$ by assumption, $r_t \perp 
\varepsilon_{1t}$ and conditional independence holds. However, this would fail 
if $r_t$ depended on $x_t$ (and hence on $\varepsilon_{1t}$), for instance if 
$r_t = f(x_t) + \varepsilon_{3t}$.

In the special case where $\Omega_t$ is binary, as in Example~\ref{Exstate} 
where $\Omega_t=S_{t-1}\in\{0,1\}$, $CAR_h(\delta,\omega)$ takes two values 
$CAR_h(\delta,0)$ and $CAR_h(\delta,1)$, corresponding to the impulse 
responses in each state. Each can be estimated by applying the NLO cross-fitting estimator of 
Section~\ref{sec:estimation} separately to each subsample 
$\{t:\Omega_t=\omega\}$, $\omega\in\{0,1\}$.

When $\Omega_t$ is continuous, as in Example~\ref{ex:cloyne} where 
$\Omega_t=r_t$, $CAR_h(\delta,\omega)$ is a function of a continuous 
argument $\omega$. From \eqref{eq:ident_CAR}, it equals the conditional 
expectation of $g_{0,h}(x_t+\delta,\mathbf{z}_{t-1},\omega)-g_{0,h}(x_t,
\mathbf{z}_{t-1},\omega)$ given $\Omega_t=\omega$, which must be estimated 
nonparametrically as a function of $\omega$. This introduces an additional 
nonparametric estimation step beyond what is required for $\mathrm{ARF}_h(\delta)$. We leave a formal treatment of this case for future work.

\section{Simulation results}\label{sec:simulations}
This section evaluates the finite sample performance of the
semiparametric local projection estimator developed in
Section~\ref{sec:estimation}. We focus on the  nonlinear regressors model of Example~\ref{Exsignsize} and the state-dependent design of
Example~\ref{ex:state}. In both cases, the structural shocks $\varepsilon_{1t}$ and $\varepsilon_{2t}$ are
mutually independent, each drawn i.i.d.\ from a Student-$t$ distribution
with $\nu=10$ degrees of freedom rescaled to unit variance, so that the
size of the shock can be interpreted as a one-standard-deviation
shock.\footnote{With $\delta=1$ and the use of a $t$-distribution for generating $\varepsilon_{1t}$, the population Riesz representer
$\alpha_0(x)$ is uniformly bounded, satisfying
Assumption~\ref{Assumption:boundedness}(i).
}
Horizons range from $h = 0, 1, \ldots, 6$ for the first DGP where the IRFs are less persistent, whereas we set $h = 0, 1, \ldots, 8$ for the second DGP. The true impulse responses are computed by simulation from the structural model
using $5{,}000$ counterfactual paths after a burn-in period of $T_0 = 500$. The DR-NLO estimator uses $K=10$ ($K=20$) contiguous blocks of equal length $\lfloor T/K \rfloor$, as described in Section~\ref{sec:estimation} for the state-dependent (nonlinear regressor) model. In both cases, $x_t$ is assumed predetermined and $x_{t-1}$ and $y_{t-1}$ are used as control variables. To estimate the conditional mean $g_{0,h}$, we use a nonparametric series estimator
based on a Hermite polynomial with the total degree selected by the AIC up to a maximum value of $2$.\footnote{While Theorem 5.1 is stated for general nonparametric estimators whose complexity grows to satisfy the rate conditions in Assumption 3, our numerical implementation utilizes a low-complexity approximation to manage the bias-variance tradeoff.}
The Riesz representer $\alpha_0$ is estimated by the LASSO minimum-distance
procedure of \citet{chernozhukov2022}, using the same
Hermite dictionary at a fixed total degree of $2$. HAC standard errors for the DR-NLO estimator are computed from the cross-fitted influence function as
described in Section~\ref{sec:estimation}, using a Bartlett kernel and  the \citet{andrews1991} automatic plug-in bandwidth. To prevent a small number of rare explosive Monte Carlo replications from dominating the results, we symmetrically trim the top and bottom 1\% of the simulated IRF estimates.
\subsection{Simulations for sign model}
This section summarizes the simulation results for a design similar to Example~\ref{Exsignsize} given by\vspace{-0.2cm}
\begin{align*}
          x_t&=0.3x_{t-1}+0.3y_{t-1}+\varepsilon_{1t}\\
           y_t&=0.5 x_{t}+0.3x_{t-1}+0.5 y_{t-1} -0.4\max(x_{t},0)-0.3\max(x_{t-1},0)+\varepsilon_{2t}.
\end{align*}\vspace{-0.2cm}

\begin{figure}[H]
    \centering    
    \caption{Performance of the DR-NLO estimator for  the nonlinear regressor DGP}
    \label{fig:nonlin}
    \includegraphics[width=0.95\textwidth]{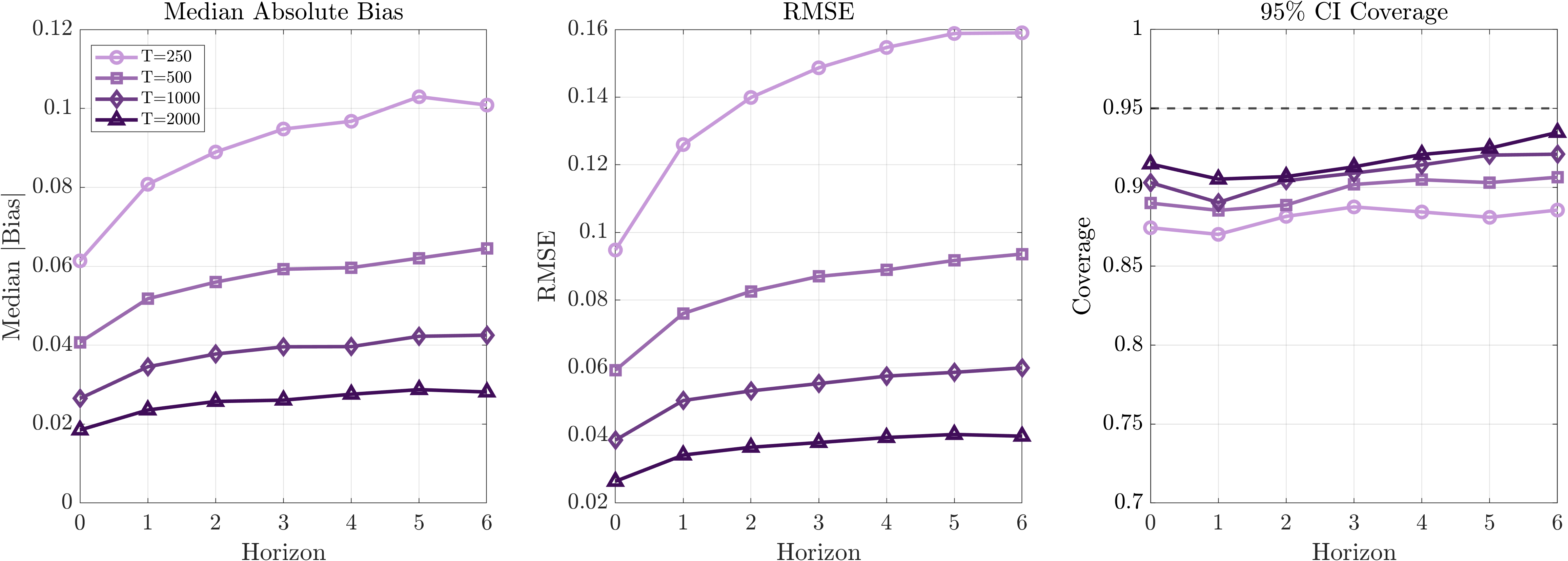}
     \caption*{\footnotesize Note: Median Absolute Bias, RMSE and Coverage for different sample sizes.} 
\end{figure}

For the sake of brevity and because simulations evaluating the performance of the plug-in and LP estimators in models with nonlinear regressors can be found in \cite{goncalves2021, goncalves2024b}, we focus on the DR-NLO estimator. We use $20,000$ replications per sample size, letting $T\in\{250,500,1000,2000\}$. For each sample size, we report the median absolute bias, root mean squared error (RMSE) and nominal coverage of the 95\% confidence intervals. The simulation results reported in Figure~\ref{fig:nonlin} indicate that DR-NLO performs well when the object of interest is the ARF in models with nonlinear regressors. The bias is small and close to zero across horizons and sample sizes, indicating that the estimator is approximately unbiased. RMSE decreases with $T$, as expected for a consistent estimator. Although coverage falls below the nominal 95\%, it improves with increasing sample size. Overall, the results suggest that the estimator performs reasonably well in modestly large samples.

\subsection{Simulations for state dependent model}
We report simulation results for the state dependent model in Example~\ref{ex:state}, where $x_t$ is predetermined so that $x_t=\rho_x x_{t-1}+\psi_x y_{t-1}+\varepsilon_{1 t}$ and $S_{t-1} = \mathbf{1}\{y_{t-1} > 0\}$ classifies the previous
period as an expansion ($S_{t-1}=1$) or recession ($S_{t-1}=0$).
The parameters are set to $\rho_x=0.3$, $\psi_x=-0.1$, $\beta_E=2.5$, $\beta_R=3.5$,$\gamma_E=0.9$, and $\gamma_R=-0.1$.   For each replication, we apply each estimator separately to the expansion
($S_{t-1}=1$) and recession ($S_{t-1}=0$) sub-samples. 
We compare our estimator to the widely used state-dependent local projection (SD-LP), implemented by OLS with state-by-covariate interactions and one lag of $(x,y)$ as controls. The Monte Carlo design uses $30{,}000$ replications per sample
size.  For each
estimator, sample size, state (expansion and recession) and horizon, we report the empirical
median bias, RMSE and the coverage of nominal $95\%$
confidence intervals across replications.\footnote{In this example, we report median bias instead of median absolute bias in order to highlight the direction of the bias, namely that the LP estimator is negatively biased.}

\begin{figure}[htpb!]
    \centering
    \caption{Median Bias}
    \label{fig:medbias_t10}
    \includegraphics[width=0.85\textwidth]{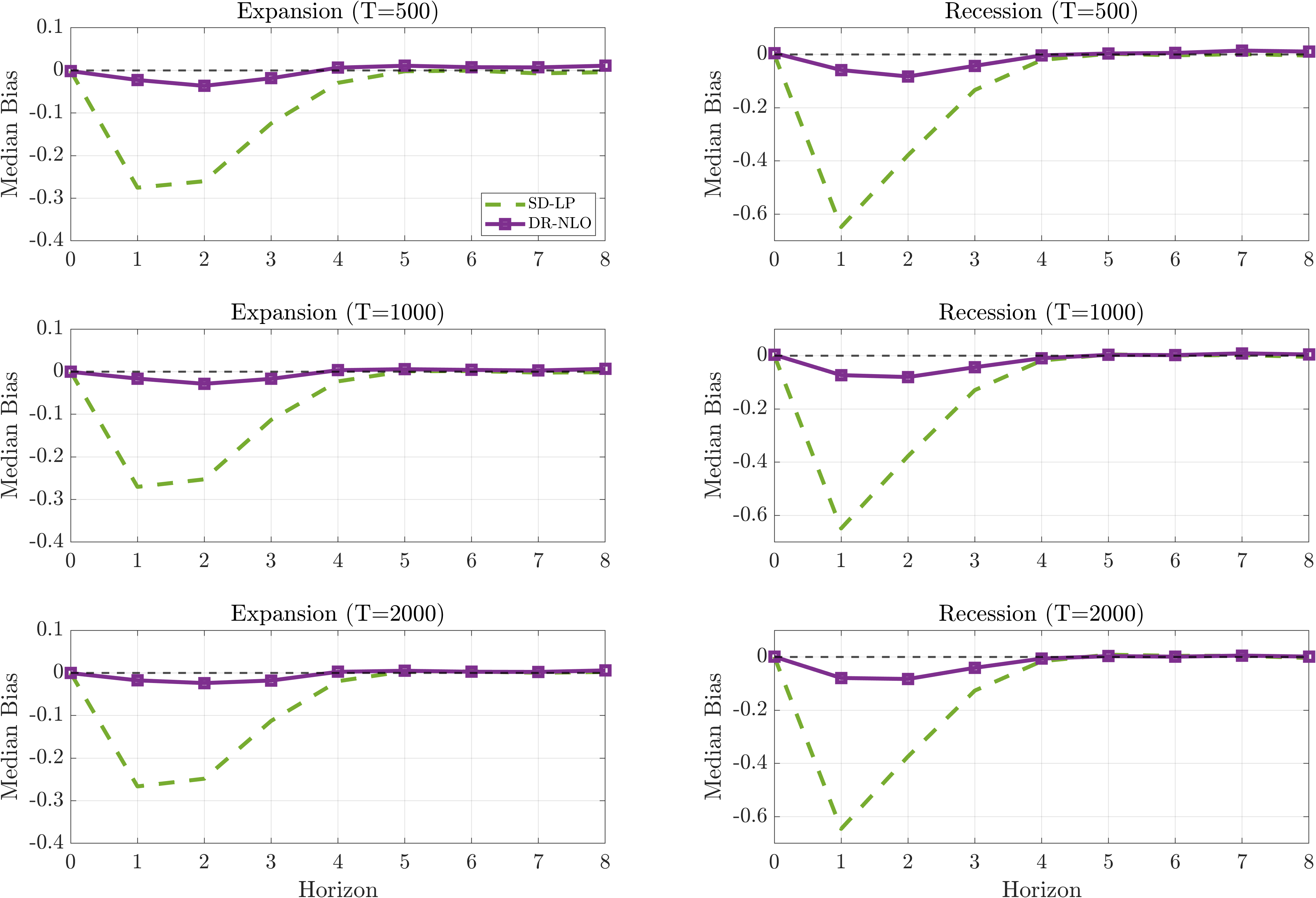}
    \caption*{\footnotesize Note: Median bias of DR-NLO and LP estimators for different sample sizes and $\delta=1$.} 
\end{figure}

\begin{figure}[htpb]
    \centering
    \caption{Root Mean Squared Error (RMSE)}
    \label{fig:RMSE_t10}
       \includegraphics[width=0.85\textwidth]{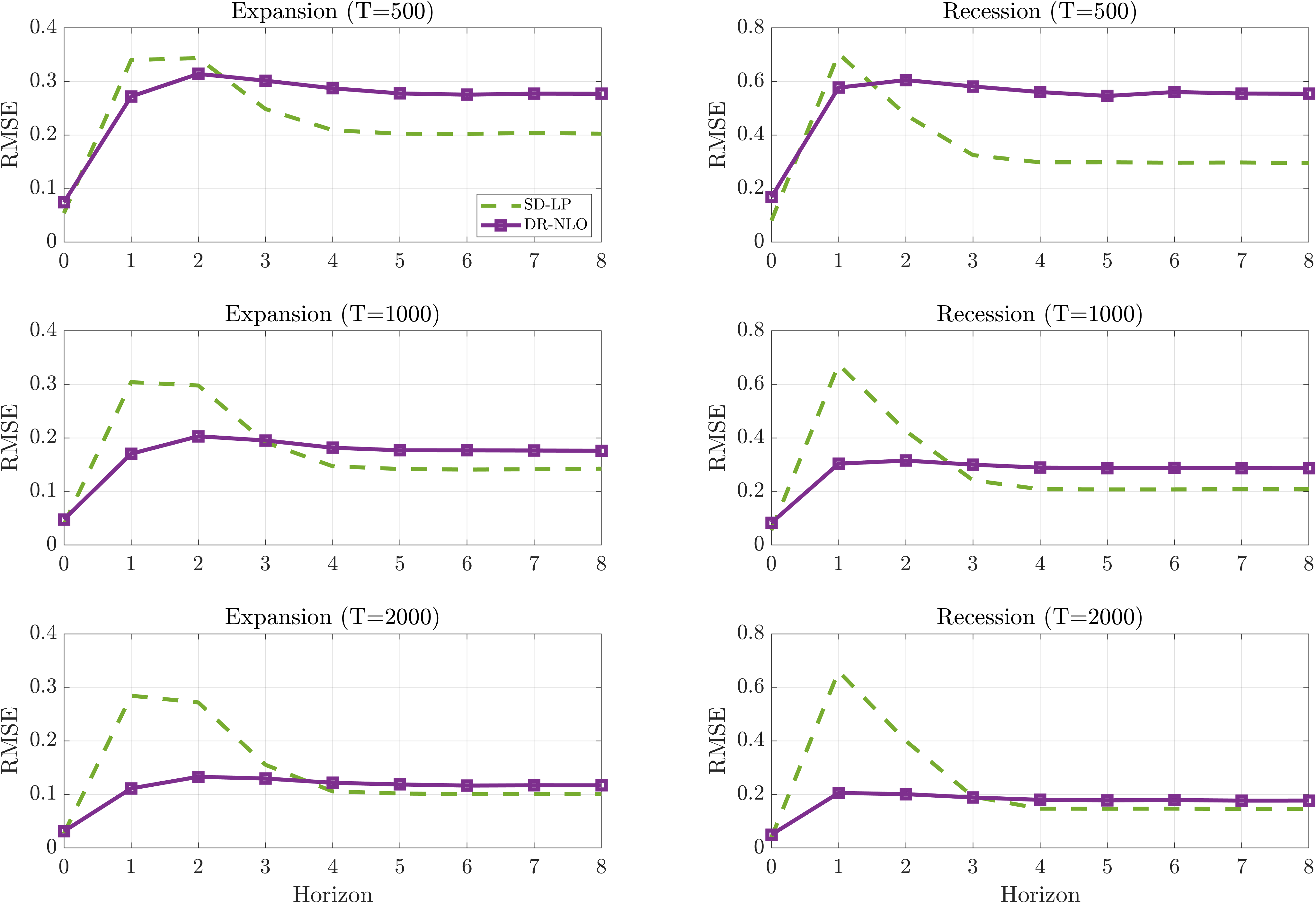}
     \caption*{\footnotesize Note: RMSE of DR-NLO and LP estimators for different sample sizes and $\delta=1$.} 
\end{figure}

\begin{figure}[htpb]
    \centering
    \caption{Empirical coverage rates of 95\% confidence intervals}
     \label{fig:Coverage_t10}
     \includegraphics[width=0.85\textwidth]{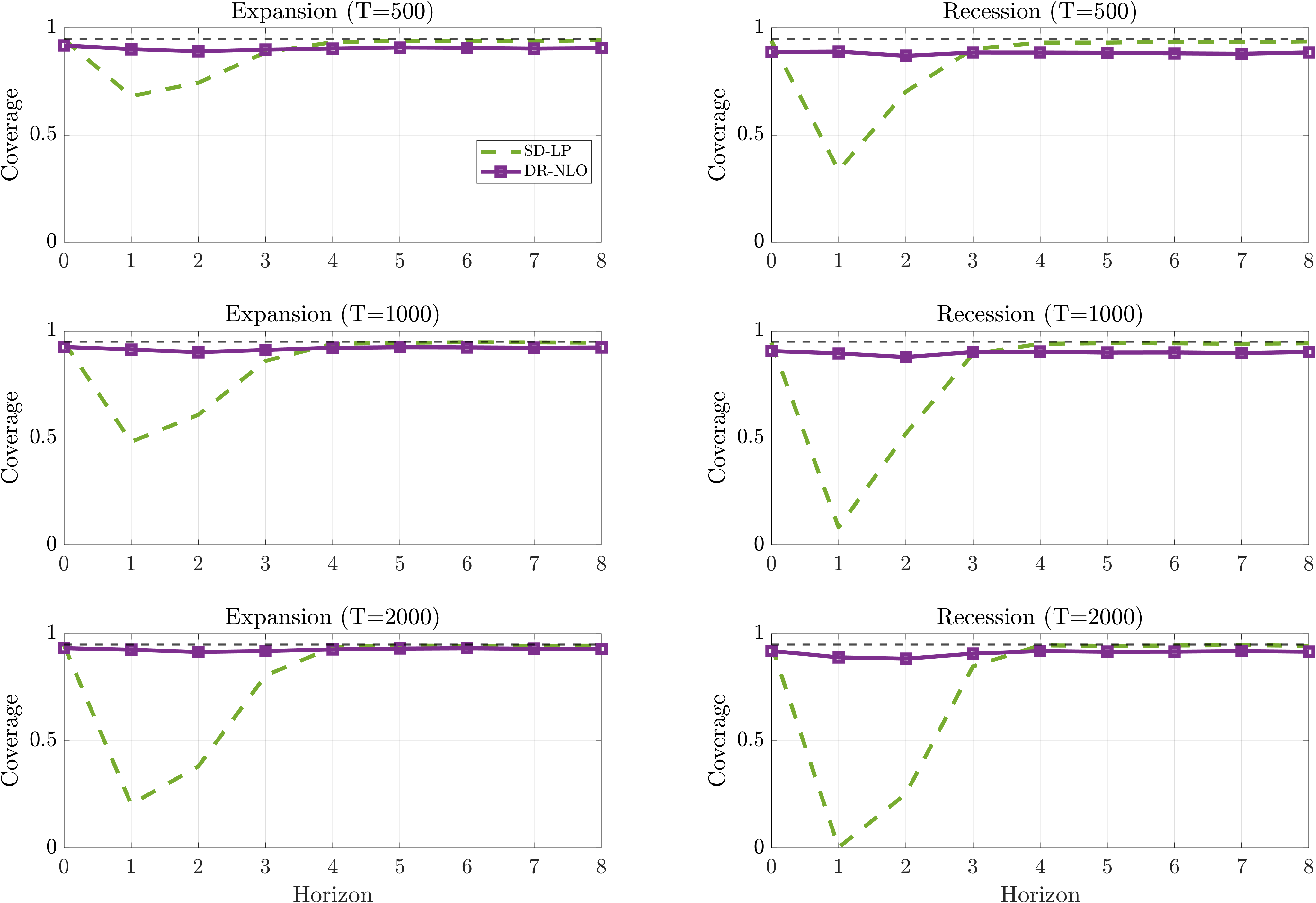}
     \caption*{\footnotesize Note: Coverage rates of 95\% intervals based on DR-NLO and SD-LP for different sample sizes and $\delta=1$.} 
\end{figure}

Figure~\ref{fig:medbias_t10} illustrates the systematic downward bias of the traditional state-dependent LP (SD-LP) estimator in both expansions and recessions, which does not vanish as the sample size grows. As explained in \citet{goncalves2024a}, when the state is endogenous, the SD-LP specification does not recover the conditional average response to a fixed size shock because state-dependent dynamics in $(\beta_{t-1},\gamma_{t-1})$ introduce nonlinearities that are not accounted for by the regression. This is not a small sample issue, but a feature of the SD-LP itself. By contrast, the semiparametric DR-NLO estimator has smaller median bias at every horizon and for every sample size considered, and its bias decreases with $T$, as predicted by Theorem~\ref{Theorem:NLO2}. The substantial bias of SD-LP at shorter horizons explains why its RMSE is greater than that of DR-NLO, despite the latter having slightly 
higher variance. At longer horizons, the bias of the SD-LP declines, resulting in a lower RMSE. As $T$ increases, differences in RMSE between the two estimators vanish at longer horizons, with DR-NLO dominating SD-LP at shorter horizons.

Figure~\ref{fig:Coverage_t10} shows the empirical coverage probabilities 
of 95\% confidence intervals based on SD-LP and DR-NLO, both using HAC 
variance estimators with a Bartlett kernel and Andrews (1991)'s automatic 
bandwidth. The SD-LP intervals exhibit substantial undercoverage for all 
sample sizes, which worsens as $T$ increases since the variance shrinks 
while the bias persists. DR-NLO slightly undercovers as well, but its 
coverage remains closer to the nominal 95\% level and improves with $T$.

 \section{Empirical illustrations}\label{sec:empirical}

Semiparametric LP estimators such as the DR-NLO estimator are useful not
only for capturing nonlinearities when the functional form of the
nonlinearity is unknown, as illustrated by our simulation evidence, but also
as a diagnostic tool for judging the adequacy of linear approximations. To
illustrate how our semiparametric LP estimator may be used to assess the
adequacy of the linear LP estimator, we apply both methods to study the
pass-through from retail gasoline price shocks to inflation in the United
States.\footnote{The reason the recent literature has focused on retail gasoline prices
rather than the price of crude oil is that the relationship between crude
oil and retail gasoline prices is unstable over time, reflecting large
variation in the cost share of crude oil in the retail price of gasoline
(see \citet{kilian2022oil}).} This has been a question of continued policy
interest, especially in recent years. There has been a proliferation of
research addressing this question using linear VAR and distributed lag
models (e.g., \citet{chudik2022estimation}; \citet{kilian2022oil, kilian2025oil}). The
use of a semiparametric LP estimator is natural in this context since it has
long been suspected that this pass-through may be nonlinear. A particular
concern is that the pass-through may be stronger when the pre-existing level
of inflation is higher. This concern has become particularly relevant in
recent months, with the surge in headline inflation following the outbreak
of the Iran War in late February 2026, but similar concerns already arose during the 2022 surge in inflation.

While this question has been addressed by a number of studies, such as \citet{clark2010time}, \citet{grundler2024does} or \citet{de2025energy}, these
studies are based on parametric nonlinear models such as threshold VAR
models, regime-switching models, or time-varying coefficient VAR models. Two
obvious concerns are that these specifications are mutually exclusive and
that they do not exhaust the range of possible nonlinear specifications. Our
semiparametric approach allows us to dispense with these parametric
restrictions. We allow the effect of gasoline price shocks ($\varepsilon
_{1t}$) on inflation to depend on lagged headline and core inflation,
according to the model:\vspace{-0.2cm}
\begin{eqnarray*}
x_{t} &=&\phi (\mathbf{z}_{t-1})+\varepsilon _{1t} \\
y_{1t} &=&\mu _{1}(x_{t},\mathbf{z}_{t-1},\varepsilon
_{2t}) \\
y_{2t} &=&\mu _{2}(x_{t},\mathbf{z}_{t-1},\varepsilon
_{3t}),
\end{eqnarray*}
where $\mathbf{z}_{t-1}$ contains six lags of the percent
change in retail gasoline prices ($x_t$), headline inflation ($y_{1t}$) and core CPI inflation for all
urban consumers ($y_{2t}$). Core inflation is defined as inflation excluding food and
energy. All data are monthly and seasonally adjusted. The data source is
FRED. The model incorporates six lags consistent with other recent empirical
studies. We are interested in the inflation
responses at horizon $0,1,\ldots,12$. The estimation sample spans January 1974
through April 2026, which includes several high-inflation episodes.

We first compare the linear LP and DR-NLO estimates of the unconditional 
responses of headline and core inflation to a one percent gasoline price 
shock, using the same learners as in the Monte Carlo simulations and $K=10$. Evidence
that the semiparametric LP estimates are very different from the linear LP estimates would cast doubt on prior linear estimates of the
pass-through to inflation. Evidence that the two estimates are close, in
contrast, would reassure policymakers that existing LP and VAR estimates
based on linear approximations are informative. Figure~\ref{fig:empap1} illustrates that
the choice of the estimator matters little. It shows point estimates of the
unconditional impulse responses and 68\% and 90\% pointwise confidence bands. Both
estimates indicate that headline inflation jumps by close to 0.05 percentage
points (not annualized) in response to a one percentage point gasoline price
shock, but the response declines quickly. By horizon 2, it is close to zero.
There is no evidence that headline inflation rises persistently in response
to gasoline price shocks. The response of core inflation is an order of
magnitude smaller and only slightly positive. Both the linear LP and the DR-NLO estimate are marginally statistically significant at horizons 1 and 2. Overall, this evidence suggests
that linear approximations are adequate for assessing the unconditional
response of inflation to gasoline price shocks.\footnote{%
Very similar results would have been obtained based on a bivariate model for
gasoline prices and headline inflation.}
\begin{figure}[htpb]
    \centering
    \caption{Average Responses of Headline and Core Inflation to Gasoline Price Shocks}
    \label{fig:empap1}
    \includegraphics[width=0.8\textwidth]{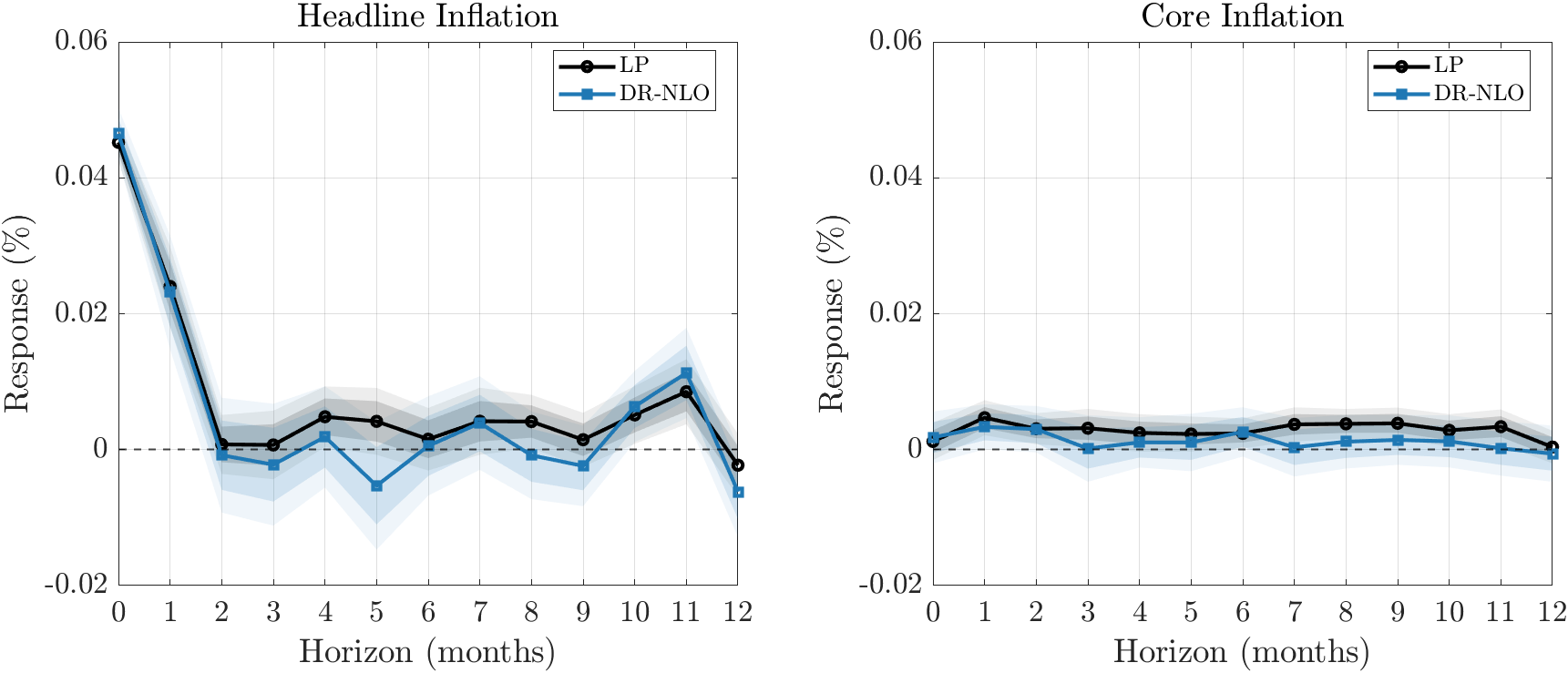}
     \begin{minipage}{0.8\textwidth}
        \footnotesize\textit{Notes:}  This figure reports LP and DR-NLO estimates of headline and core inflation to a one percentage point ($\delta=1$) increase in the price of gasoline. The shaded areas represent 68\% and 90\% confidence intervals.
    \end{minipage}
\end{figure}

Next, we turn to the question of whether there is evidence of larger
inflation responses conditional on the annualized inflation exceeding its historical
average of 3.8\% (i.e., 0.32\% monthly) during the estimation period. Figure~\ref{fig:empap2} shows the
unconditional response and the response conditional on inflation being above
average based on the DR-NLO estimator. There is no material change in the
headline and core inflation response estimates at horizons 0, 1 and 2. At
longer horizons, the conditional responses tend to be larger, but also less
precisely estimated. Even a user of the conditional response estimate would
be unable to conclude that there are statistically significant increases in
inflation in response to gasoline price shocks, however. Nor are the
conditional point estimates consistent with gasoline price shocks causing
persistent increases in headline or core inflation of a material magnitude.

\begin{figure}[htpb!]
    \centering
    \caption{Comparison of Average and Conditional Average Responses of Headline and Core Inflation to Gasoline Price Shocks}
    \label{fig:empap2}
    \includegraphics[width=0.8\textwidth]{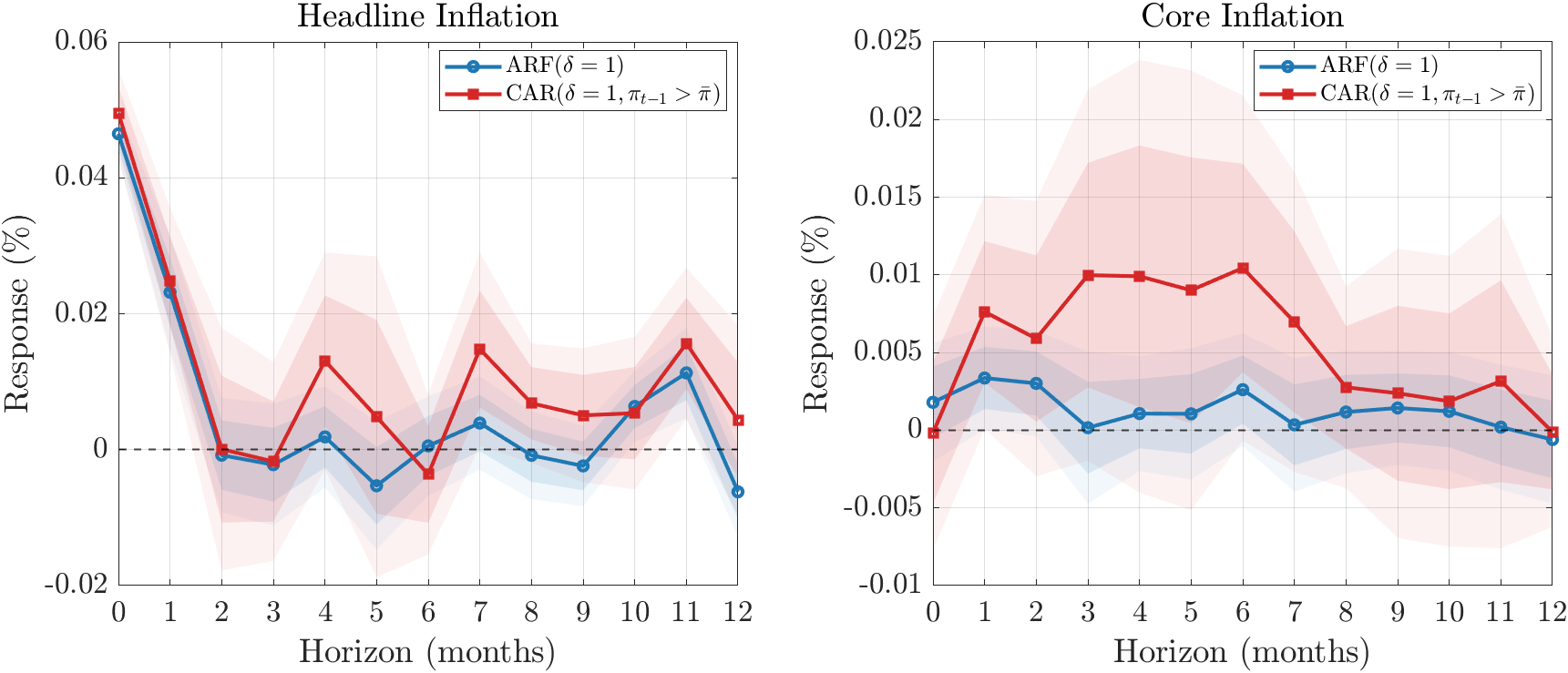}
       \begin{minipage}{0.8\textwidth}
        \footnotesize\textit{Notes:} This figure reports DR-NLO estimates of the unconditional average responses of headline and core inflation to a one percentage point gasoline price shock ($ARF_h(\delta)$ with $\delta=1$) and the corresponding average responses conditional on headline inflation being above its historical average in $t-1$ ($CAR_h(\delta, \pi_{t-1}> \bar{\pi}$) with $\delta=1$). The shaded areas represent 68\% and 90\% confidence intervals.
    \end{minipage}
\end{figure}

We conclude with another example that illustrates that the impulse responses implied by the DR-NLO estimator may look materially different from those implied by the linear LP estimator and more economically plausible. The question of interest is how sales of U.S. motor vehicles respond to shocks to the real price of motor gasoline. This question is motivated by \citet{edelstein2009sensitive} and \citet{ramey2006declining, ramey2011oil} who document potential nonlinearities in the response of motor vehicle sales. We rely on equations (\ref{eq:S1}) and (\ref{eq:S2}) with $z_{t}$ containing the percent change in real U.S. retail gasoline prices and the percent change in U.S. sales of automobiles and light trucks. Retail gasoline prices are obtained from the CPI and deflated by the aggregate CPI. All data are monthly and seasonally adjusted. The data source is FRED. The model incorporates six lags. The estimation period is January 1974 through March 2026.

\begin{figure}[htpb!]
    \centering
    \caption{Response of Motor Vehicle Sales to Real Gasoline Price Shocks ($\delta=1$ s.d.)}
    \label{fig:MVsales}
    \includegraphics[width=0.45\textwidth]{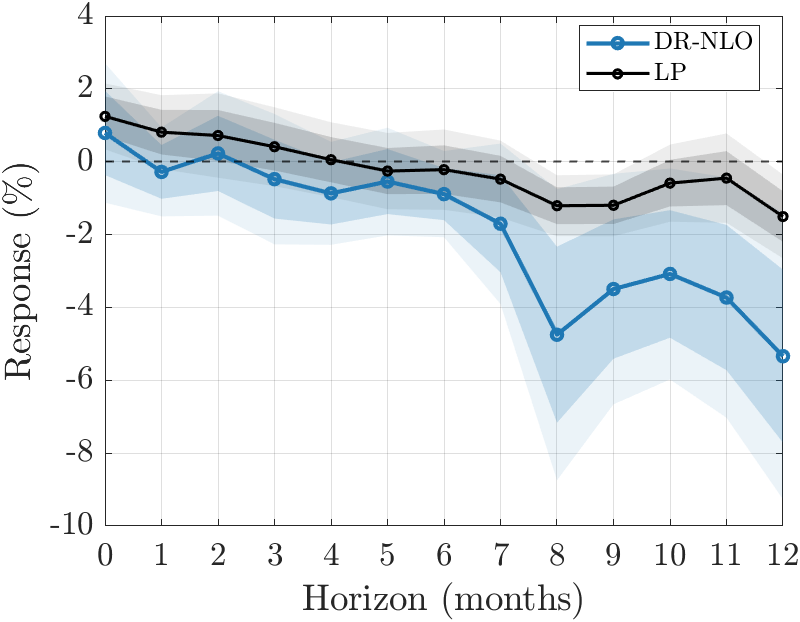}
    \begin{minipage}{0.8\textwidth}
        \footnotesize\textit{Note:} This figure reports LP and DR-NLO estimates of the cumulative response of motor vehicle sales to one standard deviation shock in the real price of gasoline. The shaded areas represent 68\% and 90\% confidence intervals.
    \end{minipage}
\end{figure}

Figure \ref{fig:MVsales} focuses on the responses of the cumulative growth rate to a one standard deviation shock in the real price of gasoline.\footnote{The cumulative effect is estimated directly by applying the LP or DR-NLO to  $\Delta^hy_{t+h}=y_{t+h}-y_{t-1}$.} The DR-NLO impact response is positive but statistically insignificant. There is no statistically significant response in sales for the first three months, according to the DR-NLO estimator, followed by a persistent decline in auto sales starting at horizon 3. The linear LP estimate, in contrast, suggests a statistically significant but economically counterintuitive increase in auto sales at horizons 0, 1 and 2. More generally, the persistently positive linear LP responses for the first four months are difficult to reconcile with economic reasoning. The numerical differences between the estimates compound at longer horizons. At horizon 12, the linear LP estimate shows a decline in auto sales that is only 28\% of the DR-NLO response. These differences suggest that linear approximations may not adequately capture the dynamics in question.

\section{Concluding remarks}\label{sec:conclusion}
This paper developed a semiparametric local projection estimator of unconditional nonlinear impulse response functions for a broad class of structural dynamic models that are widely used in applied macroeconomics, including models with nonlinearly transformed regressors, state-dependent coefficients, and nonlinear interactions between shocks and state variables. 
Under standard mixing and rate conditions on the nuisance estimators, the resulting estimator is $\sqrt{T}$-consistent
and asymptotically normal, with the preliminary estimation of the two nuisance functions having no effect on the first-order asymptotic distribution. Inference is conducted via a HAC long-run variance estimator applied to the estimated influence function. We also showed how our
framework accommodates conditional impulse responses in state-dependent models, where the conditioning variable takes discrete values.

Our Monte Carlo evidence indicates that the proposed estimator delivers substantially smaller bias than standard state-dependent local projection methods, while preserving competitive root mean squared errors and confidence-interval coverage close to the nominal level. We considered two empirical illustrations. The first one focused on the potentially nonlinear pass-through from gasoline price shocks to inflation. The second example examined whether linear LP estimators miss nonlinearities in the transmission of gasoline price shocks to motor vehicle sales.

There are several avenues for further research. First, our identification result relies on the additive separability of $\varepsilon_{1t}$ in equation \eqref{eq:S1}. One possible extension would be to allow the shock to interact nonlinearly with lagged state variables. Second, our treatment of conditional impulse responses assumes a discrete conditioning variable,
leaving the continuous-state case as discussed in Example~\ref{ex:inter} for future work.

\singlespacing
\bibliographystyle{agsm}
\bibliography{References}
\doublespacing
\setcounter{equation}{0}
\section*{Appendix} 
\appendix\setcounter{section}{0}
\section{Proof of results in Section~\ref{sec:ident}}\label{sec:proof_prop}
\begin{proof} [Proof of Proposition~\ref{Prop:identification}] By \eqref{eq:S1} and \eqref{eq:S2}, we can rewrite $y_{t+h} = m_h(x_t,\, U_{t+h})$,
for some function $m_h$ that maps $x_t$ (which we observe) and $U_{t+h}$ (defined as in the text) into $y_{t+h}$. Although $x_t$ is not independent of $U_{t+h}$ (unless $x_t=\varepsilon_{1t}$), the independence condition $\varepsilon_{1t}\perp U_{t+h}$ implies the conditional
independence assumption $x_t \perp U_{t+h} \mid \mathbf{z}_{t-1}$. This is the key identifying condition: conditionally on $\mathbf{z}_{t-1}$, all
remaining variation in $x_t$ comes from $\varepsilon_{1t}$ alone, which is
independent of $U_{t+h}$. Because the model for $x_t$ given in \eqref{eq:S1} is additive in $\varepsilon_{1t}$, a
$\delta$-shift in $x_t$ holding $\mathbf{z}_{t-1}$ fixed is identical to a $\delta$-shift
in $\varepsilon_{1t}$, so that $\theta_{0,h}\equiv\mathrm{ARF}_h(\delta)$ in Definition 1 can
equivalently be written as
$\theta_{0,h}=
  E\!\left[m_h(x_t + \delta,\, U_{t+h}) - m_h(x_t,\, U_{t+h})\right].
  \label{eq:ARF2}$
But for any fixed $x$ and $z$, $g_{0,h}(x,z)\equiv E(y_{t+h}\mid x_t=x,\,\mathbf{z}_{t-1}=z) = E\!\left[m_h(x,\, U_{t+h})\mid\mathbf{z}_{t-1}=z\right]$
where the second equality uses the fact that $U_{t+h}$ is independent of $x_t$, conditionally on $\mathbf{z}_{t-1}$. The desired result follows by applying the law of iterated expectations (LIE).
\end{proof}

\begin{proof}[Proof of Proposition~\ref{Prop:DR}]
By \eqref{eq:Neyman} with $g_h=g_{0,h}$, we have 
$E[g_{0,h}(x_t+\delta,\mathbf{z}_{t-1})-g_{0,h}(x_t,\mathbf{z}_{t-1})]
=E[\alpha_0(x_t,\mathbf{z}_{t-1})g_{0,h}(x_t,\mathbf{z}_{t-1})]=\theta_{0,h}$.
Substituting $\theta_h=\theta_{0,h}$ in \eqref{eq:moment_DR} and using 
$E[\alpha_0(x_t,\mathbf{z}_{t-1})(y_{t+h}-g_{0,h}(x_t,\mathbf{z}_{t-1}))]=0$, 
which follows by LIE since $E[y_{t+h}-g_{0,h}(x_t,\mathbf{z}_{t-1})|x_t,\mathbf{z}_{t-1}]=0$,
gives the result.
\end{proof}

\begin{proof}[Proof of Proposition~\ref{Prop:identification_CAR}]
The proof follows by the same arguments as the proof of 
Proposition~\ref{Prop:identification}, with $\mathbf{z}_{t-1}$ replaced 
by $(\mathbf{z}_{t-1},\Omega_t)$ throughout, and applying the law of iterated expectations (LIE) conditioning 
additionally on $\Omega_t=\omega$.
\end{proof}
\section{Proofs of results in Section~\ref{sec:estimation} }\label{sec:proofs}
To prove our results, we rely on Lemma A.6 of \citet{semenova2023}. For completeness, we describe this result next. 

\begin{lemma}\label{lem:huang5}
Let $A(z_t,\eta)$ be some generic function of data $z_t$ and nuisance 
function $\eta\in \mathcal{H}_T$, where $\mathcal{H}_T$ is a shrinking neighborhood of $\eta_0$ which contains 
the estimator $\hat{\eta}_\ell$ with probability converging to one for all 
$\ell=1,\ldots,K$. Suppose that $\{z_t\}$ is geometrically $\beta$-mixing and define
\[
B_{\ell}(\eta)\equiv \frac{1}{T_\ell}\sum_{t\in I_\ell}E[A(z_t,\eta)]
\quad \text{and}\quad 
V_{\ell}(\eta)\equiv \frac{1}{T_\ell}\sum_{t\in I_\ell}[A(z_t,\eta)-E(A(z_t,\eta))].
\]
If for any non-stochastic sequence $\eta_T\in \mathcal{H}_T$ and any norm $\|\cdot\|$,
\begin{flalign}
\|B_\ell(\eta_T)\|&=O(\xi_{1T})\tag{A.1}\label{eq:A1} \\
\|V_\ell(\eta_T)\|&=O_p(\xi_{2T}),\tag{A.2}\label{eq:A2} 
\end{flalign}
then 
$
\Bigl\|T^{-1}\sum_{\ell=1}^{K}\sum_{t\in I_\ell}A(z_t,\hat{\eta}_\ell)\Bigr\|
=O_p(\xi_{1T}+\xi_{2T}).
$
\end{lemma}
\begin{proof}[Proof of Theorem~\ref{Theorem:NLO2}]
Let $\tilde{\theta}_h$ denote the oracle estimator that uses the true conditional mean function $g_{0,h}(x_t,\mathbf{z}_{t-1})$ and the true Riesz representer $\alpha_0(x_t,\mathbf{z}_{t-1})$. We decompose
\[
\sqrt{T}(\hat{\theta}_h-\theta_{0,h}) = \sqrt{T}(\tilde{\theta}_h-\theta_{0,h})+ \sqrt{T}(\hat{\theta}_h-\tilde{\theta}_h).
\]
The first term satisfies $\sqrt{T}(\tilde{\theta}_h-\theta_{0,h})\to_d N(0,V_h)$ by a CLT for weakly dependent time series (e.g., Corollary 24.7 of \citet{davidson1994}, since $\beta$-mixing implies strong mixing), using Assumptions~\ref{Assumption:mixing} and \ref{Assumption:boundedness}(iii). It remains to show that $\sqrt{T}(\hat{\theta}_h-\tilde{\theta}_h)=o_p(1)$.

Writing $\hat{\theta}_h-\tilde{\theta}_h=R_1+R_2+R_3$ with
\begin{flalign*}
  R_1 &\equiv \frac{1}{T}\sum_{\ell=1}^{K}\sum_{t\in I_\ell}\bigl(m(x_t,\mathbf{z}_{t-1},\hat{g}_\ell-g_0)+\alpha_0(x_t,\mathbf{z}_{t-1})(g_0(x_t,\mathbf{z}_{t-1})-\hat{g}_\ell(x_t,\mathbf{z}_{t-1}))\bigr), \\
  R_2 &\equiv \frac{1}{T}\sum_{\ell=1}^{K}\sum_{t\in I_\ell}(\hat{\alpha}_\ell(x_t,\mathbf{z}_{t-1})-\alpha_0(x_t,\mathbf{z}_{t-1}))(y_{t+h}-g_0(x_t,\mathbf{z}_{t-1})), \\
  R_3 &\equiv \frac{1}{T}\sum_{\ell=1}^{K}\sum_{t\in I_\ell}(\hat{\alpha}_\ell(x_t,\mathbf{z}_{t-1})-\alpha_0(x_t,\mathbf{z}_{t-1}))(g_0(x_t,\mathbf{z}_{t-1})-\hat{g}_\ell(x_t,\mathbf{z}_{t-1})),
\end{flalign*}
where $m(x_t,\mathbf{z}_{t-1},\hat{g}_\ell-g_0)\equiv (\hat{g}_\ell(x_t+\delta,\mathbf{z}_{t-1})-\hat{g}_\ell(x_t,\mathbf{z}_{t-1}))-(g_0(x_t+\delta,\mathbf{z}_{t-1})-g_0(x_t,\mathbf{z}_{t-1}))$, we show that $\sqrt{T}R_j=o_p(1)$ for $j=1,2,3$. Since $\{z_t\}$ is geometrically $\beta$-mixing (Assumption~\ref{Assumption:mixing}) and $I_\ell$ and $I^{\mathrm{qc}}_\ell$ are separated by at least $T_\ell$ time periods, the hypothesis of Lemma~\ref{lem:huang5} is satisfied, and it suffices to verify conditions~(\ref{eq:A1}) and (\ref{eq:A2}) for non-stochastic sequences $\eta_T\in\mathcal{G}_T\times\mathcal{A}_T$. Without loss of generality we take $\ell=1$, so that $t\in I_\ell$ corresponds to $t=1,\ldots,T_1$; we nonetheless retain the generic index $\ell$ in the sums below. We define $\xi_t\equiv (x_t,\mathbf{z}_{t-1})$ throughout.

We start with $R_1$ and apply Lemma~\ref{lem:huang5} with 
$$
A(\xi_{t},\eta)=m(x_t,\mathbf{z}_{t-1},g-g_0)+\alpha_0(x_t,\mathbf{z}_{t-1})(g_0(x_t,\mathbf{z}_{t-1})-g(x_t,\mathbf{z}_{t-1})),
$$ 
where $\eta=g$. By Neyman orthogonality, $E[A(\xi_{t},\eta_T)]=0$ for any $g_T\in\mathcal{G}_T$, so $B_\ell(\eta_T)=0$ and condition~(\ref{eq:A1}) holds trivially. For condition~(\ref{eq:A2}), we show $\mathrm{Var}(V_\ell(\eta_T))=o(T_\ell^{-1})$. By stationarity of $z_t$ (and hence of $\xi_{t}$),
\[
\mathrm{Var}(V_\ell(\eta_T)) = \frac{1}{T_\ell}E[A(\xi_{t},\eta_T)^2] + \frac{2}{T_\ell^2}\sum_{j=1}^{T_\ell-1}(T_\ell-j)\,\mathrm{Cov}(A(\xi_{t},\eta_T),A(\xi_{t-j},\eta_T)).
\]
For the first term, we have
\begin{flalign*}
E[A(\xi_{t},\eta_T)^2]\le 2\left(E[m(x_t,\mathbf{z}_{t-1},g_T-g_0)^2]+E[\alpha_0(x_t,\mathbf{z}_{t-1})^2(g_T(x_t,\mathbf{z}_{t-1})-g_0(x_t,\mathbf{z}_{t-1}))^2]\right).
\end{flalign*}
Using Assumption~\ref{Assumption:boundedness}(i) and a change of variables, 
\begin{flalign*}
  &E[m(x_t,\mathbf{z}_{t-1},g_T - g_0)^2]\\
  &\le 2 \left(\int \left(\int (g_T(x_t,\mathbf{z}_{t-1})-g_0(x_t,\mathbf{z}_{t-1}))^2 \left(\frac{f_0(x_t-\delta|\mathbf{z}_{t-1})}{f_0(x_t|\mathbf{z}_{t-1})}+1\right)f_0(x_t|\mathbf{z}_{t-1})dx_t\right)f_0(\mathbf{z}_{t-1})d\mathbf{z}_{t-1}\right)\\
  &\le 2(\bar{\alpha}+2)r_{g,T}^2=o(1),
\end{flalign*}
where $f_0(\mathbf{z}_{t-1})$ denotes the marginal density of $\mathbf{z}_{t-1}$, the last inequality uses $\sup_{x,z}|\alpha_0(x,z)|<\bar\alpha$ and the last equality uses Assumption~\ref{Assumption:learners}(i). An analogous argument gives $E[\alpha_0(x_t,\mathbf{z}_{t-1})^2(g_T(x_t,\mathbf{z}_{t-1})-g_0(x_t,\mathbf{z}_{t-1}))^2]\le\bar\alpha^2 r_{g,T}^2=o(1)$, so $T_\ell^{-1}E[A(\xi_{t},\eta_T)^2]=o(T_\ell^{-1})$. For the second term, $A(\xi_{t},\eta_T)$ inherits the geometric $\beta$-mixing of $\{z_t\}$, and since $\beta$-mixing implies strong mixing, by Corollary 14.3 of \citet{davidson1994}, for $q>2$ as in Assumption~\ref{Assumption:boundedness}(ii),
\[
|\mathrm{Cov}(A(\xi_{t},\eta_T),A(\xi_{t-j},\eta_T))|\le C\,\beta(j)^{1-2/q}\|A(\xi_{t},\eta_T)\|_q^2.
\]
By the triangle inequality and Assumption~\ref{Assumption:boundedness}(i), $|A(\xi_{t},\eta_T)|\le |m(x_t,\mathbf{z}_{t-1},g_T-g_0)|+\bar\alpha\,|g_T(x_t,\mathbf{z}_{t-1})-g_0(x_t,\mathbf{z}_{t-1})|$, so $\|A(\xi_{t},\eta_T)\|_q \le C(\bar\alpha)\,r_{g,q,T} = o(1)$ by Assumption~\ref{Assumption:learners}(i). Since $\beta(j)\le Ce^{-c_\beta j}$ with $c_\beta>0$, summing over $j$ gives
\[
\frac{2}{T_\ell^2}\sum_{j=1}^{T_\ell-1}(T_\ell-j)|\mathrm{Cov}(A(\xi_{t},
\eta_T),A(\xi_{t-j},\eta_T))|\le \frac{C(\bar\alpha)\,r_{g,q,T}^2}{T_\ell}
\sum_{j=1}^{\infty}e^{-c_\beta j(1-2/q)}=o(T_\ell^{-1}),
\]
where $\sum_{j=1}^{\infty}e^{-c_\beta j(1-2/q)}<\infty$ because $c_\beta(1-2/q)>0$ 
for $c_\beta>0$ and $q>2$, and $r_{g,q,T}^2=o(1)$ by 
Assumption~\ref{Assumption:learners}(i). Combining the two bounds gives $\mathrm{Var}(V_\ell(\eta_T))=o(T_\ell^{-1})$, so $\sqrt{T}R_1=o_p(1)$.

We next consider $R_2$. We apply Lemma~\ref{lem:huang5} with $A(\xi_{t},\eta_T)=(\alpha_T(x_t,\mathbf{z}_{t-1})-\alpha_0(x_t,\mathbf{z}_{t-1}))e_{t+h}$, where $e_{t+h}\equiv y_{t+h}-g_0(x_t,\mathbf{z}_{t-1})$. Since $E(e_{t+h}|x_t,\mathbf{z}_{t-1})=0$, an application of the LIE gives $E[A(\xi_{t},\eta_T)]=0$, so $B_\ell(\eta_T)=0$ and condition~(\ref{eq:A1}) holds trivially. For condition~(\ref{eq:A2}), let $f_t\equiv(\alpha_T(x_t,\mathbf{z}_{t-1})-\alpha_0(x_t,\mathbf{z}_{t-1}))e_{t+h}$. By stationarity,
$
\mathrm{Var}(V_\ell(\eta_T)) = \frac{1}{T_\ell}E(f_t^2) + \frac{2}{T_\ell^2}\sum_{j=1}^{T_\ell-1}(T_\ell-j)\,\mathrm{Cov}(f_t,f_{t-j}).
$
For the variance, by LIE and Assumption~\ref{Assumption:boundedness}(ii),
\[
E(f_t^2)=E\left[(\alpha_T(x_t,\mathbf{z}_{t-1})-\alpha_0(x_t,\mathbf{z}_{t-1}))^2 E(e_{t+h}^2|x_t,\mathbf{z}_{t-1})\right]
\le\bar{\sigma}_q^2\,r_{\alpha,T}^2=o(1),
\]
by Assumption~\ref{Assumption:learners}(i) (since $r_{\alpha,T}=r_{\alpha,2,T}\le r_{\alpha,q,T}$ by H\"older's inequality), so $T_\ell^{-1}E(f_t^2)=o(T_\ell^{-1})$. For the autocovariance terms, $f_t$ is a function of $(y_{t+h},\xi_t)$ and hence inherits the geometric $\beta$-mixing of $\{z_t\}$. By the same covariance inequality as above, $|\mathrm{Cov}(f_t,f_{t-j})|\le C\,\beta(j)^{1-2/q}\|f_t\|_q^2.$ By LIE and Assumption~\ref{Assumption:boundedness}(ii), $E[|f_t|^q]\le\bar\sigma_q^q\,r_{\alpha,q,T}^q$, so $\|f_t\|_q\le\bar\sigma_q\,r_{\alpha,q,T}=o(1)$ by Assumption~\ref{Assumption:learners}(i). Summing over lags gives
\[
\frac{2}{T_\ell^2}\sum_{j=1}^{T_\ell-1}(T_\ell-j)|\mathrm{Cov}(f_t,f_{t-j})|\le\frac{C\,\bar\sigma_q^2\,r_{\alpha,q,T}^2}{T_\ell}\sum_{j=1}^{\infty}e^{-c_\beta j(1-2/q)}=o(T_\ell^{-1}),
\]
by Assumption~\ref{Assumption:learners}(i). Combining both bounds gives $\mathrm{Var}(V_\ell(\eta_T))=o(T_\ell^{-1})$, so $\sqrt{T}R_2=o_p(1)$.

Finally, we consider $R_3$. We apply Lemma~\ref{lem:huang5} with $$A(\xi_{t},\eta)=(\alpha(x_t,\mathbf{z}_{t-1})-\alpha_0(x_t,\mathbf{z}_{t-1}))(g(x_t,\mathbf{z}_{t-1})-g_0(x_t,\mathbf{z}_{t-1}))$$ where $\eta=(\alpha,g)$. For the bias term, by Cauchy-Schwarz,
\[
|E[A(\xi_{t},\eta_T)]|\le \bigl(E[(\alpha_T(x_t,\mathbf{z}_{t-1})-\alpha_0(x_t,\mathbf{z}_{t-1}))^2]\bigr)^{1/2}\bigl(E[(g_T(x_t,\mathbf{z}_{t-1})-g_0(x_t,\mathbf{z}_{t-1}))^2]\bigr)^{1/2},
\]
which is bounded by $r_{\alpha,T}\,r_{g,T}=o(T^{-1/2})$, given Assumption~\ref{Assumption:learners}(ii). So $B_\ell(\eta_T)=o(T^{-1/2})$ and condition~(\ref{eq:A1}) is satisfied. For the stochastic term, by stationarity, the triangle inequality, and Cauchy-Schwarz,
$
E|V_\ell(\eta_T)|\le 2E|A(\xi_{t},\eta_T)|\le 2r_{\alpha,T}\,r_{g,T}=o(T^{-1/2}),
$
by Assumption~\ref{Assumption:learners}(ii), so condition~(\ref{eq:A2}) holds by Markov's inequality. This completes the proof that $\sqrt{T}R_3=o_p(1)$.
\end{proof}
\end{document}